\pdfoutput=1
\documentclass[10pt,a4paper]{article}

\usepackage{amsmath,amsgen,latexsym}
\usepackage{amstext,amssymb,amsfonts,latexsym}
\usepackage{theorem}
\usepackage{pifont}
\usepackage{graphicx}

 \newcommand{\bs}{\bigskip}
 \newcommand{\ms}{\medskip}
 \newcommand{\n}{\noindent}
 \newcommand{\s}{\smallskip}
 \newcommand{\hs}[1]{\hspace*{ #1 mm}}
 \newcommand{\vs}[1]{\vspace*{ #1 mm}}

 \newcommand{\setempty}{\varnothing}
 
 \newcommand{\nat}{\mathbb{N}}
 
 \newcommand{\integer}{\mathbb{Z}}

 \newcommand{\CC}{{\cal C}}

 \newcommand{\LL}{{\cal L}}

 \newcommand{\TT}{{\cal T}}
 \newcommand{\PP}{{\cal P}}
 \newcommand{\QQ}{{\cal Q}}

 \newcommand{\dl}{\mathrm{L}}
 
 \newcommand{\p}{\mathrm{P}}

 \newcommand{\fl}{\mathrm{FL}}

 \newcommand{\cfl}{\mathrm{CFL}}
 
 \newcommand{\dcfl}{\mathrm{DCFL}}

\theoremstyle{plain}
\theoremheaderfont{\bfseries}
 \newtheorem{theorem}{Theorem}[section]
 \newtheorem{lemma}[theorem]{Lemma}
 \newtheorem{proposition}[theorem]{{\bf Proposition}}
 \newtheorem{corollary}[theorem]{Corollary}

 {\theorembodyfont{\rmfamily}
  \newtheorem{definition}[theorem]{Definition}}
 {\theorembodyfont{\rmfamily} }
 {\theorembodyfont{\rmfamily} }

 \newenvironment{proof}{\par \noindent
            {\bf Proof. \hs{2}}}{\hfill$\Box$ \vspace*{3mm}}

 \newenvironment{yproof}{\par \noindent
            {\bf Proof. \hs{2}}}{\hfill$\Box$ \vspace*{3mm}}

 \newenvironment{proofof}[1]{\vspace*{5mm} \par \noindent
         {\bf Proof of #1.\hs{2}}}{\hfill$\Box$ \vspace*{3mm}}

 \newcommand{\ceilings}[1]{\lceil #1 \rceil}
 
 \newcommand{\pair}[1]{\langle #1 \rangle}

\newcommand{\ignore}[1]{}

\newcommand{\track}[2]{[\: \begin{subarray}{c} #1 \\%
      #2 \end{subarray} ]}

\newcommand{\dollar}{\$}

\newcommand{\logcfl}{\mathrm{LOGCFL}}
\newcommand{\logdcfl}{\mathrm{LOGDCFL}}

\newcommand{\Lmreduces}{\leq^{\mathrm{L}}_{m}}
\newcommand{\ksda}[1]{#1 \mathrm{SDA}}

\newcommand{\logksda}[1]{\mathrm{LOG} #1 \mathrm{SDA}}
\newcommand{\logksna}[1]{\mathrm{LOG} #1 \mathrm{SNA}}
\newcommand{\auxsna}{\mathrm{AuxSNA}}
\newcommand{\mmid}{\!\mid\!}

\begin{document}
\pagestyle{plain}
\pagenumbering{arabic}
\setcounter{page}{1}
\setcounter{footnote}{0}

\begin{center}
{\Large {\bf Nondeterministic Auxiliary Depth-Bounded Storage Automata and Semi-Unbounded Fan-in Cascading Circuits}}\footnote{This current article extends and corrects its preliminary report \cite{Yam22} that has appeared in the Proceedings of the 28th International Computing and Combinatorics Conference (COCOON 2022), Shenzhen, China, October 22--24, 2022,
Lecture Notes in Computer Science, vol. 13595, pp. 61--69, Springer, 2022. The conference talk was given online because of the coronavirus pandemic.}
\bs\s\\

{\sc Tomoyuki Yamakami}\footnote{Present Affiliation: Faculty of Engineering, University of Fukui, 3-9-1 Bunkyo, Fukui 910-8507, Japan}\\
\end{center}
\ms

\begin{abstract}
We discuss a nondeterministic variant of the recently introduced  machine model of deterministic auxiliary depth-$k$ storage automata (or aux-$k$-sda's) by Yamakami. It was proven that all languages recognized by polynomial-time logarithmic-space aux-$k$-sda's are located between $\mathrm{LOGDCFL}$ and $\mathrm{SC}^k$ (the $k$th level of Steve's class SC).
We further propose a new and simple computational model of semi-unbounded fan-in Boolean circuits composed partly of cascading blocks, in which the first few AND gates of unbounded fan-out (called AND$_{(\omega)}$ gates) at each layer from the left (where all gates at each layer are indexed from left to right) are linked in a ``cascading'' manner to their right neighbors though specific AND and OR gates. We use this new circuit model to characterize a nondeterministic variant of the aux-$2k$-sda's (called aux-$2k$-sna's) that run in polynomial time using logarithmic work space.
By relaxing the requirement for cascading circuits, we also demonstrate how such cascading circuit families characterize the complexity class $\p$. This yields an upper bound on the computational complexity of  $\mathrm{LOG}k\mathrm{SNA}$ by $\p$.

\s
\n{\bf Keywords.}
{auxiliary storage automata, LOG$k$SNA, semi-unbounded fan-in circuit, cascading circuit, frozen blank sensitivity, NSC$^{k}$}
\end{abstract}

\sloppy
\section{Background and Main Contributions}\label{sec:introduction}

We quickly review necessary background knowledge on parallel computation, introduce central computational models, and describe our main contributions associated with these models.

\subsection{The Language Families $k$SDA and LOG$k$SDA}

Opposed to sequential computation, \emph{parallel computation} has gained great interests in many real-life situations.
Such parallel computation has been realized by numerous computational models, typically alternating Turing machines (or ATMs, for short) and circuit families, in the past literature. As such examples, uniform families of polynomial-size $O(\log^kn)$-depth Boolean circuits characterize the complexity class $\mathrm{AC}^k$ and polynomial-time ATMs with $O(\log^kn)$ alternations induce the complexity class $\mathrm{SC}^k$.

Here, let us zero in on the complexity class, known as $\logcfl$, which is originally defined by Cook \cite{Coo71} as the closure of $\cfl$ (context-free language class) under logarithmic-space many-one reductions (or $\dl$-m-reductions, for short).
As Sudborough \cite{Sud78} demonstrated, this complexity class is precisely characterized by two  restricted forms of machine-based models, including auxiliary pushdown automata and multi-head pushdown automata.
These models still rely on the use of \emph{(pushdown) stacks} and corresponding stack operations, where a ``stack'' is a quite simple memory device whose access is limited to the topmost stored data and it is disallowed to read/write the data stored in the other part of the device.
Ruzzo \cite{Ruz80} used ATMs to precisely characterize languages in $\logcfl$.
The simulation of stack behaviors plays a crucial role in those  characterizations because a series of stack operations is, intuitively,  related to depth-first search of a computation tree of an ATM.
Venkateswaran \cite{Ven91} analyzed Ruzzo's simulation techniques and  further took a step to give a precise characterization of languages in $\logcfl$ in terms of uniform families of semi-unbounded fan-in circuits.
Notice that circuit characterizations are in general quite useful for a better understanding of parallelism of computations for given languages. Similar characterizations were obtained  for the counting variant of $\logcfl$ in \cite{Vin91} and also for the unambiguous variant of $\logcfl$ in \cite{NR95}.
The deterministic variant of $\logcfl$, denoted $\logdcfl$, is also characterized by PRAMs \cite{DR86,FLR96} and multiplex circuits \cite{FLR96,MRV99}. It is important to note that $\logdcfl$ is included in $\mathrm{SC}^2$ \cite{Coo79}.

To expand $\cfl$ as well as $\dcfl$ (deterministic context-free language class), numerous models have been proposed in the past literature by simply allowing a direct access to more  stored data in the entire memory device. By allowing more accesses to stacks, for instance, Ginsburg, Greibach, and Harrison \cite{GGH67a,GGH67b} studied  \emph{stack automata} and Meduna \cite{Med06} introduced \emph{deep pushdown automata}.
Our particular interest lies on Hibbard's \cite{Hib67} \emph{deterministic $k$-limited automaton} (abbreviated as $k$-lda), which uses a single read/write tape for reading inputs and storing information but the access to each tape cell is severely limited to the first $k$ visits. Such a rewriting restriction may possibly model, for example, physical abrasions of computer's hard drives caused by frequent frictions between the surface of the drive and a scanning head.
It turns out that the ``rewriting'' of tape cells provides enormous power to computation although the number of rewriting opportunities is limited.
This makes the language families defined in terms of $k$-lda's for various numbers $k$ form a proper infinite hierarchy in between $\dcfl$ and  $\cfl$ \cite{Hib67}. In sharp contrast, the nondeterministic variant of $k$-lda's recognize only context-free languages \cite{Hib67}.
As for a recent progress on a study of $k$-limited automata, the reader refers to, e.g., \cite{PP14,Yam19}.

In 2021, a new machine model, called \emph{(one-way) deterministic depth-$k$ storage automata}\footnote{It is possible, as noted in \cite{Yam21}, to differentiate between two variations of $k$-sda's by way of \emph{allowing} or \emph{disallowing} the input-tape head to move around while scanning any  frozen and near frozen storage tape cells. The former condition introduces an ``depth-immune'' model and the latter with no tape head moves on a specific frozen blank symbol introduces a ``depth-susceptible'' model. See \cite{Yam21} for their definitions and properties.} (or $k$-sda's, for short), was proposed in \cite{Yam21} as a natural extension of deterministic pushdown automata as well as Hibbard's $k$-lda's.
In its formulation, a separate \emph{storage tape} plays a key role as a functional extension of a (pushdown) stack, by allowing its tape head to revise the content of the  storage tape only during the first $k$ visits. This is intended to handle more complex operations over stored items in the memory device.
The roles of an input tape and a storage tape are clearly different. The input tape is read only, maintaining given fixed data whereas the storage tape is rewritable, storing and freely modifying stored data.
The family of languages recognized by $k$-sda's, denoted  $\ksda{k}$, naturally contains the language families induced by (one-way) deterministic pushdown automata and Hibbard's $k$-lda's.
Moreover, the non-context-free language $L_{abc}=\{a^nb^nc^{2n}\mid n\geq1\}$ belongs to  $4\mathrm{SDA}$; on the contrary, $2\mathrm{SDA}$\footnote{In Lemma 2.1 of \cite{Yam21}, the depth-susceptible model of 2-sda's are used to claim this assertion. In contrast, this work presents another characterization of $\dcfl$ by 2-sda's with a new constraint of \emph{(frozen) blank sensitivity} in Section \ref{sec:FBS}.} coincides with $\dcfl$ \cite{Yam21}.
Similar to $\logdcfl$ induced from $\dcfl$, for each index $k\geq1$, $\logksda{k}$ is defined to be composed of all languages that are $\dl$-m-reducible to languages in $\ksda{k}$.
This family $\logksda{k}$  was shown in \cite{Yam21} to be located between $\logdcfl$ and $\mathrm{SC}^k$ (the $k$th level of Steve's class SC) for any $k\geq2$.
It was proven also in \cite{Yam21} that all languages in $\logksda{k}$ are precisely characterized in terms of   \emph{(two-way) deterministic auxiliary depth-$k$ storage automata} (or aux-$k$-sda's) running in polynomial time using logarithmic (auxiliary) work space.
This characterization was proven via another multi-head machine model, called \emph{two-way multi-head deterministic depth-$k$ storage automata}.

Unfortunately, many characteristic features of storage automata still remain elusive from our understandings. For example, we do not know whether  $\logksda{(k+1)}$ is different from $\logksda{k}$ or whether $\logksda{k}$ is properly contained in $\mathrm{SC}^k$. To promote our understanding of storage automata, we nevertheless need to explore further features of these exquisite machines and their variants.

\subsection{Nondeterministic Variants: $k$-sna's and aux-$k$-sna's}\label{sec:contribution}

The nondeterministic variants of $k$-lda's, called $k$-lna's, was also studied by Hibbard \cite{Hib67} by allowing underlying machines to make a nondeterministic choice at every step. Unlike the deterministic case of $k$-lda's, nevertheless,
the recognition power of $k$-lna's does not exceed that of nondeterministic pushdown automata.
As an instant consequence, the family of all languages recognized by $k$-lna's are neither closed under intersection nor complementation.  In the case of storage automata,  it is possible to analogously define  \emph{(one-way) nondeterministic depth-$k$ storage automata} (or  $k$-sna's) from $k$-sda's by requiring that, for every computation path generated by each $k$-sna, the number of rewriting each tape cell is upper-bounded by $k$.
Therefore, the rewriting restriction severely hinders the computational power. Different from $k$-lna's, $k$-sna's are  more powerful than pushdown automata.
For their precise definitions, refer to Section \ref{sec:k-sda's}.
In natural analogy to $k\mathrm{SDA}$ and $\logksda{k}$, we intend to use the notation $k\mathrm{SNA}$ for the class of all languages recognized by $k$-sna's and the notation $\logksna{k}$ for the closure of $k\mathrm{SNA}$ under logspace reductions. Naturally, we can ask a question of what the computational complexity of $\logksna{k}$ is. By the nondeterminization of $k\mathrm{SDA}$, we obtain the following statement.

\renewcommand{\labelitemi}{$\circ$}
\begin{enumerate}\vs{-1}
\item[(1)] For any positive integer $k$, $\logksna{k}$ is located between $\logcfl$ and $\mathrm{NSC}^k$, where $\mathrm{NSC}^k$ is the nondeterministic version of $\mathrm{SC}^k$.
     (Proposition  \ref{LOGkSNA-NSC})
\end{enumerate}

Similarly to aux-$k$-sda's, we introduce the computational model of \emph{nondeterministic auxiliary depth-$k$ storage automata} (or aux-$k$-sna's). Now, let us introduce the generic notation of $\auxsna\mathrm{depth,\!space,\!time}(k,s(n),t(n))$ for a family of languages  using aux-$k$-sna's that run in $O(t(n))$ time using $O(s(n))$ space.
In this work, we particularly focus on the frozen blank symbol $B$, which must be placed on storage-tape cells after the first $k$ visits to these tape cells and is never replaced afterward.
Later, we will assume that every aux-$k$-sna satisfies a condition on the behavior of input-tape heads while reading $B$, called \emph{(frozen) blank sensitivity}.

For the languages in $\logcfl$, Venkateswaran \cite{Ven91} demonstrated an elegant characterization of them  in terms of  uniform families of polynomial-size $O(\log{n})$-depth semi-unbounded fan-in circuits, where a family of  \emph{semi-unbounded fan-in circuits} refers to a circuit family for which there is a constant $\ell\geq1$ such that any path from the root to a leaf node in any circuit in the family has at most $\ell$ consecutive gates of AND. Similar circuit characterizations were obtained by Vinay \cite{Vin91} and Niedermeier and Rossmanith \cite{NR95} for unambiguous and counting variants of $\logcfl$.
As a natural question, we ask whether the circuit characterization of languages in $\logcfl$ can be extended to $\logksna{k}$.

This  work intends to seek out an alternative characterization  of aux-$k$-sna's  in terms of certain forms of Boolean circuits. To establish the desired characterization, we introduce a modified form of Boolean circuits, which contain special subcircuits made of AND gates of ``unbounded fan-out'' (denoted AND$_{(\omega)}$) together with standard AND and OR gates.
Such a subcircuit is specifically called a  \emph{cascading block}. See Section \ref{sec:cascading} for its formal definition.
Cascading blocks with cascading length at most $k$ are succinctly called $k$-cascading blocks.

In particular, we consider logarithmic-depth polynomial-size semi-unbounded fan-in circuits built up partly with $k$-cascading blocks. These special circuits are dubbed as \emph{$k$-cascading circuits}.
The \emph{alternation} of such a circuit is the maximum number of times when the gate types switch between AND and OR (ignoring AND$_{(\omega)}$ gates) along any path from the output gate to an input gate. The \emph{size} of a circuit is the total number of gates and wires (which connect between gates) in it.
The notation $\mathrm{CCIRcasc,\!alt,\!size}(k,O(\log{n}),n^{O(1)})$ expresses the collection of all languages computed by families of such circuits of semi-unbounded fan-in, polynomial-size, and logarithmically-bounded alternations under the log-space uniformity condition.

Of all key results of this work, we will prove the following.

\renewcommand{\labelitemi}{$\circ$}
\begin{enumerate}\vs{-1}
\item[(2)]
For any positive integer $k$,  $\auxsna\mathrm{depth,\!space,\!time}(2k,O(\log{n}),n^{O(1)})$   coincides with $\mathrm{CCIRcasc,\!alt,\!size}(k,O(\log{n}),n^{O(1)})$, provided that all aux-$2k$-sna's are (frozen) blank sensitive. (Theorem   \ref{circuit-character})
\end{enumerate}

The upper bound $k$ of cascading length is an important parameter in characterizing $\auxsna\mathrm{depth,\!space,\!time}(2k,O(\log{n}),n^{O(1)})$ by $\mathrm{CCIRcasc,\!alt,\!size}(k,O(\log{n}),n^{O(1)})$. By removing this upper bound, we further obtain the characterization of $\p$.

\renewcommand{\labelitemi}{$\circ$}
\begin{enumerate}\vs{-1}
\item[(3)]
The complexity class $\p$ coincides with $\mathrm{CCIRcasc,\!alt,\!size}(n^{O(1)}, n^{O(1)}, n^{O(1)})$. (Proposition \ref{P-cascading-circuit})
\end{enumerate}

The aforementioned three results will be proven in Sections \ref{sec:simulation} and \ref{sec:P-connection} after explaining our basic computational models in Section \ref{sec:basics}.

\section{Basic Machine Models: $k$-sna's and aux-$k$-sna's}\label{sec:basics}

We briefly explain our machine models used in the rest of this work: nondeterministic storage automata and nondeterministic auxiliary storage automata.

\subsection{Numbers, Sets, and Languages}\label{sec:numbers-sets}

We use two notations $\integer$ and $\nat$ for the set of all \emph{integers} and that of all \emph{natural numbers} (i.e., nonnegative integers), respectively.
Moreover, another notation $\nat^{+}$ is used for the set $\nat-\{0\}$.
An \emph{integer interval} $[m,n]_{\integer}$ refers to the set $\{m,m+1,m+2,\ldots,n\}$ for two arbitrary numbers $m,n\in\integer$ with $m\leq n$. We further abbreviate $[1,n]_{\integer}$ as $[n]$ if $n\geq1$.
In this work, all \emph{polynomials} are assumed to have natural-number coefficients and all \emph{logarithms} are taken to the base $2$.
Given a set $A$, $\PP(A)$ denotes the \emph{power set} of $A$, i.e., the set of all subsets of $A$.

A finite nonempty set of letters or symbols is called an \emph{alphabet} and a finite sequence of such symbols from an alphabet $\Sigma$ is called a \emph{string} over $\Sigma$. The \emph{length} of a string $x$,  denoted $|x|$, is the total number of symbols in $x$. A unique string of length $0$ is called the \emph{empty string} and always denoted $\varepsilon$. The set of all strings over $\Sigma$ is expressed as $\Sigma^*$. A \emph{language} over $\Sigma$ is just a subset of $\Sigma^*$. Given a number $n\in\nat$, $\Sigma^n$ denotes the set of all strings of length exactly $n$. Notice that $\Sigma^0 = \{\varepsilon\}$.

For a string $x$ of length $n$ over alphabet $\Sigma$ and for any number $i\in[n]$, the notation $x_{(i)}$ denotes the $i$th symbol of $x$. Thus, the sequence $x_{(1)}x_{(2)}\cdots x_{(n)}$ coincides with $x$. For our convenience, we further set $x_{(0)}=\varepsilon$ and $x_{(n+i)}=x_{(n)}$ for any integer $i\geq0$.

\subsection{Auxiliary Storage Automata with Depth-Bounded Storage Tapes}\label{sec:k-sda's}

Let us explain two computational models of storage automata and auxiliary storage automata, which are central to this work.
To make the subsequent definitions concise, we assume the reader's familiarity with nondeterministic auxiliary pushdown automata (or aux-npda's, for short) introduced in \cite{Coo71}.
Notice that aux-npda's running in polynomial time using logarithmic work space characterize exactly languages in $\logcfl$ \cite{Sud78}.


The machine model of \emph{deterministic auxiliary depth-$k$ storage automata} (or aux-$k$-sda's, for short) was first studied in \cite{Yam21}.
The interested reader may refer to \cite{Yam21} for the detailed description of this model.
In this work, however, we introduce its ``nondeterministic'' variant, called a \emph{nondeterministic auxiliary depth-$k$ storage automaton} (or an aux-$k$-sna), as a two-way 3-tape nondeterministic Turing machine  with a read-only input tape, a rewritable auxiliary (work) tape, and a rewritable depth-$k$ storage tape.
Initially, an auxiliary tape and a storage tape are all blank except for the left endmarkers and their corresponding tape heads start at these endmarkers.


We formally define an aux-$k$-sna $M$ as a nonuple $(Q,\Sigma, \{{\rhd},{\lhd}\}, \{\Gamma^{(e)}\}_{e\in[0,k]_{\integer}}, \Theta, \delta, q_0, Q_{acc}, Q_{rej})$, where $Q$ is a finite set of inner states, $\Sigma$ is an input alphabet, $\rhd$ and $\lhd$ are endmarkers, $\Gamma^{(e)}$ is the $e$th storage alphabet with $\Gamma^{(0)}=\{\Box\}$ (where $\Box$ is called the \emph{initial blank symbol}) and $\Gamma^{(k)}=\{{\rhd},B\}$ (where $B$ is the \emph{frozen blank symbol}), $\Theta$ is an auxiliary (work) tape  alphabet, $q_0$ is the initial state in $Q$, and $Q_{acc}$ and $Q_{rej}$ are respectively a set of accepting states and that of rejecting states with $Q_{acc}\cup Q_{rej}\subseteq Q$ and $Q_{acc}\cap Q_{rej}=\setempty$. The input tape has two endmarkers and both the storage and the auxiliary tapes have only the left endmarker. A transition function $\delta$ maps $(Q-Q_{halt}) \times\check{\Sigma}  \times \check{\Theta} \times \Gamma$ to $\PP(Q\times \check{\Theta} \times \Gamma \times  D^3)$ with $\check{\Sigma}=\Sigma\cup\{{\rhd},{\lhd}\}$,
$\Gamma=\bigcup_{e\in[0,k]_{\integer}} \Gamma^{(e)}$, $\check{\Theta} = \Theta\cup\{{\rhd}\}$, $Q_{halt} = Q_{acc}\cup Q_{rej}$, and $D =\{-1,0,+1\}$. For any symbol $\gamma$ in $\Gamma^{(e)}$, ``$e$'' is referred to as the \emph{depth} of $\gamma$ and denoted  $depth(\gamma)$. As customary, the endmarkers are never used to replace any tape symbol other than the endmarkers. The three head directions $-1$, $+1$, and $0$ mean ``moving to the left'', ``moving to the right'', and ``staying still'' (or stationary move), respectively.

A transition of the form ``$(p,\theta,\xi,d_1,d_2,d_3)\in\delta(q,\sigma,\tau,\gamma)$''  expresses that, in scanning three symbols $\sigma$, $\tau$, and $\gamma$ on $M$'s input, storage, and auxiliary tapes, $M$ changes its inner state from $q$ to $p$ by moving its input-tape head in direction $d_1$, writes $\theta$ over $\tau$ (possibly $\theta=\tau$) on an auxiliary (work) tape by moving its tape head in direction $d_2$, and writes $\xi$ over $\gamma$ on a storage tape by moving its tape head in direction $d_3$. In particular, when $d_3=0$ (resp., $d_1=0$), $N$'s move is conveniently referred to as a \emph{storage stationary move} (resp., an \emph{input stationary move}).
We need to distinguish between the following two types of storage stationary moves: (1) the first stationary move taken just after the storage-tape head comes to the current tape cell from its neighboring tape cell and (2) any stationary moves that follow other stationary moves. The stationary move of type (1) and type (2) are respectively called \emph{intrinsic} and \emph{extrinsic} and they lead to completely different tape head moves.


A \emph{left turn} (resp., a \emph{right turn}) of a tape head at a tape cell is a circumstance in which the tape head moves to this tape cell from the left (resp., the right) and, by a non-stationary move at the next step,
it leaves the tape cell to the left (resp., the right). A \emph{turn} refers to either a left turn or a right turn.
%
%
The frozen blank symbol $B$ is a key to this work. A \emph{frozen blank turn} means a turn of the storage-tape head while reading the frozen blank symbol $B$. The \emph{no frozen blank turn condition} states that the storage-tape head of an aux-$k$-sna $M$ does not make any frozen blank turn on any computation path of $M$ on any input.

The following conditions are collectively called the \emph{depth-$k$ requirement} of aux-$k$-sna's.

\s
\n{\sf [depth-$k$ requirement]}
When $M$ scans a symbol $\gamma\in\Gamma^{(e)}$ on a storage-tape cell with $e<k$,  if $M$ makes a turn (thus, $d_3\neq0$), then $M$ should rewrite another symbol $\xi$ in $\Gamma^{(\min\{e+2,k\})}$ over $\gamma$. Each turn made at a tape cell is viewed as ``double'' visits to this tape cell. If $M$ moves to a neighboring tape cell, then $M$ should rewrite another symbol $\xi$ in $\Gamma^{(\min\{e+1,k\})}$ over $\gamma$.
Additionally, all symbols in $\Gamma^{(k)}$ are not rewritable.

\s
\n{\sf [stationary requirement]}
In this work, we demand the following three conditions. (1) $d_1=d_2=d_3=0$ never happens; namely, at least one of
the tape heads must move at any step.
(2) $d_2=0$ implies $\theta=\tau$.
(3) In the case of $e<k$ and $d_3=0$, if $M$ makes an intrinsic stationary move, then it should revise the content of tape cell from $\gamma$ to another symbol $\xi$ in $\Gamma^{(\min\{e+1,k\})}$; if $M$ makes an extrinsic stationary move, then $M$ should maintain $\gamma$ on the current tape cell (which intuitively corresponds to the $\varepsilon$-move of a pushdown automaton).

\s


\begin{figure}[t]
\centering
\includegraphics*[height=4.8cm]{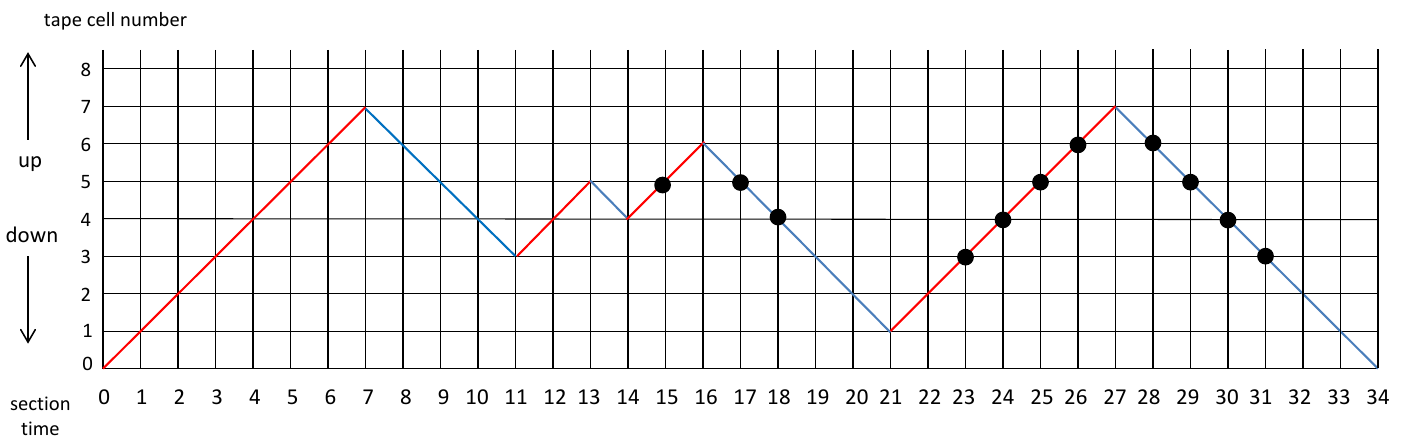}
\caption{Movement of $M$'s storage-tape head along a fixed computation path of $M$ on input $x$ with $k=4$, where each lattice point indicates a surface configuration at section time $t$ and is referred to as $C_t$.
A {section time} means the number of steps by merging a consecutive series of all stationary moves of the storage-tape head into a non-stationary move that precedes the series.
For example, the storage-tape head in $C_2,C_9,C_{11},C_{18}$ scans cells  $2,5,3,4$, respectively. Cell 4 is visited by the tape head in $C_4,C_{10},C_{12},C_{14},C_{18}$, etc. The black circles indicate that the storage-tape cells are frozen with the frozen blank symbol $B$. For example, $C_{17},C_{23},C_{26}$ contain frozen blank tape cells.}\label{fig:storage-tape-head-move}
\end{figure}


It is possible to modify the basic definition of aux-$k$-sna so that, for an appropriately fixed constant $c\in\nat^{+}$, all  storage-tape cells are partitioned into blocks of $c\ceilings{\log{n}}$ cells, which are succinctly called \emph{cell blocks}, and we force the machine to sweep each of such cell blocks from left to right or from left to right by rewriting all cells of the block, where $n$ denotes an input size. In the calculation of the storage depth, we can treat each cell block as a single tape cell. Those cell blocks will play an important role in Section \ref{sec:lemma-proof-II}.

A \emph{configuration} of $M$ on input $x$ is a sextuple $(q,l_1,l_2,w,l_3,u)$, which indicates that $M$ is in inner state $q$, scanning cell $l_1$ of the input tape and cell $l_3$ of the storage tape holding string $u$ together with the entire tape content $w$ of the auxiliary tape whose tape head rests on cell $l_2$.
The initial configuration of $M$ on input $x$ is $(q_0,0,0,\rhd\Box^{m_1},0,\rhd\Box^{m_2})$, where $m_1$ (resp., $m_2$) is the maximum number of storage-tape cells (resp., auxiliary-tape cells)  used by $M$ on all computation paths of $M$ on $x$.
When $q\in Q_{acc}\cup Q_{rej}$, the sextuple $(q,l_1,l_2,w,l_3,u)$ is called a \emph{halting configuration}.
After making a transition of the form $(p,v,b,d_1,d_2,d_3)\in \delta(q,\sigma,z,a)$, a configuration $P= (q,l_1,l_2,w,l_3,u)$ is changed to $P'= (p,l_1+d_1,l_2+d_2,w',l_3+d_3,u')$, where $\sigma$ is $x_{(l_1)}$,  $z$ is $w_{(l_2)}$, $a$ is $u_{(l_3)}$, $u'$ is obtained from $u$ by substituting $b$ for $a$, and $w'$ is obtained from $w$ by substituting $v$ for $z$.
We express this transition from $P$ to $P'$ as $P\vdash_{M} P'$.
A \emph{computation path} of $M$ on input $x$ is a series of configurations such that (i) the series begins with the initial configuration of $M$ and ends with a halting configuration (whenever the series is finite) and (ii) any two consecutive configurations $(P,P')$ in the series satisfies $P\vdash_{M} P'$.
An \emph{accepting computation path} (resp., a \emph{rejecting computation path}) refers to a computation path that terminates with an accepting configuration (resp., a rejecting configuration).

A \emph{computation} is defined as a \emph{computation graph}, which is a directed acyclic graph whose vertices are labeled by  configurations and directed edges represent single transitions between  configurations.
A path of such a computation graph from the initial configuration of $M$ on $x$ to a leaf, which is a halting configuration, thus represents a  computation path of $M$ on $x$. Given a number $t\in\nat^{+}$, we write $P\vdash_{M,t} R$ when there exists a computation path of length $t$ from $P$ to $R$ made by $M$ on $x$. We often drop ``$M$'' in $P\vdash_{M}R$ and $P\vdash_{M,t}R$, and write $P\vdash R$ and $P\vdash_{t}R$, respectively, when $M$ is clear from the context.
A \emph{storage history} describes a sequence of changes of contents of the storage tape together with locations of its tape head according to time.
To understand the movement of a storage-tape head of an aux-$k$-sna (with $k=4$) along its computation path, we include a figure depicting such a movement as Fig.~\ref{fig:storage-tape-head-move}.

It is convenient to take an entire series of storage stationary moves made between two non-stationary moves and merge it into a single non-stationary move that precede the series.
A \emph{section time} refers to the number of steps taken by $M$ by merging a entire series of consecutive stationary moves of the storage-tape head into a single non-stationary move that precedes the series.

An input string $x$ is \emph{accepted} by $M$ if $M$ starts with $x$ written on the input tape with the endmarkers and the computation graph of $M$ on $x$ contains an accepting computation path. In contrast, $x$ is  \emph{rejected} by $M$ if all finite computation paths of $M$'s computation graph on $x$ are rejecting.
A language $L$ over alphabet $\Sigma$ is \emph{recognized} by $M$ if, for all $x\in L$, $M$ accepts $x$ and, for all $x\in\Sigma^*-L$, $M$ rejects $x$. In this case, $L$ is also expressed as $L(M)$. Let $\mathrm{Time}_{M}(x)$ denote the runtime of $M$ on input $x$.

\begin{definition}
Let $s(n)$ and $t(n)$ denote any two bounding functions (i.e., functions from $\nat$ to $\nat$). The notation $\auxsna\mathrm{depth,\!space,\!time}(k,s(n),t(n))$ denotes the collection of all languages $L(M)$ recognized by aux-$k$-sna's $M$ using   $O(s(n))$ tape cells on the auxiliary tape within $O(t(n))$ steps, provided that aux-$k$-sna's should satisfy the depth-$k$ requirement.
Remember that no space restriction is imposed on the storage tape.
\end{definition}

In this work, we are particularly interested in the case of $s(n)=O(\log{n})$ and $t(n)=n^{O(1)}$.

\subsection{Frozen Blank Sensitivity and Weak Depth Susceptibility}\label{sec:FBS}

Finally, we introduce another computational model of \emph{(one-way) nondeterministic depth-$k$ storage automaton} (or $k$-sna, for short) by removing the device of auxiliary (work) tapes and by limiting the input-tape head to move in only one direction (i.e., from left to right) with additional stationary moves.
We further demand that either the input-tape head or the storage-tape head (or both) must move at any step.

To see the complexity of $k$-sna's, as a quick example, let us consider the following decision problem, called the \emph{k-XOR Problem}. For two binary strings $x$ and $y$ of length $n$, the notation $x\oplus y$ indicates the bitwise XOR  of $x$ and $y$. For example, if $x=0110$ and $y=1100$, then $x\oplus y$ equals $1010$.

\ms
{\bf $k$-XOR Problem:}
\begin{itemize}\vs{-1}
  \setlength{\topsep}{-2mm}%
  \setlength{\itemsep}{1mm}%
  \setlength{\parskip}{0cm}%

\item[] Instance: $e\# a_1\# a_2\# \cdots \# a_n$ with $n\geq1$ and $e,a_1,a_2,\ldots,a_n\in\{0,1\}^n$.

\item[] Question: Is there a series $(i_1,i_2,\ldots,i_k)$ of indices with $1\leq i_1<i_2<\cdots <i_k\leq n$ such that $e$ equals $(a_{(i_1)})^R\oplus a_{(i_2)} \oplus (a_{(i_3)})^R \oplus a_{(i_4)} \oplus \cdots \oplus (a_{(i_k)})^R$ if $k$ is odd, and $(a_{(i_1)})^R\oplus a_{(i_2)} \oplus (a_{(i_3)})^R \oplus a_{(i_4)} \oplus \cdots \oplus a_{(i_k)}$ otherwise?
\end{itemize}

This problem can be solvable by an appropriate $(k+1)$-sna, which first reads and writes $eB$, chooses $i_1,i_2,\ldots,i_k$ nondeterministically, reads $a_{(i_1)},a_{(i_2)},\ldots,a_{(i_k)}$ in order, writes down $b_1=e\oplus (a_{(i_1)})^R$, $b_2=b_1\oplus a_{(i_2)}$, $b_3=b_2\oplus (a_{(i_3)})^R$, $\ldots$ over the same block of $n$ tape cells by moving a tape head back and forth (using $\rhd$ in cell $0$ and $B$ in cell $n+1$) until $b_{k-1}$, and finally checks if the resulting string $b_{k}$ matches $0^n$.

Yamakami \cite{Yam22} studied two essentially different restrictions of $k$-sda's on the behavior of their input-tape heads and the usage of their inner states while scanning specific storage-tape symbols. In the \emph{depth-susceptible model}, (1) the input-tape head cannot move while scanning any symbol in $\Gamma^{(k-1)}\cup \Gamma^{(k)}$ and (2) an inner state cannot change while reading the frozen blank symbol $B$. In the \emph{depth-immune model}, on the contrary, the input-tape head is allowed to move around with no restriction and an inner state freely changes.
The same paper asserted in \cite[Lemma 2.1]{Yam21} that depth-susceptible $2$-sda's precisely characterize $\dcfl$ (i.e., the class of all deterministic content-free languages).


In this work, nevertheless, we wish to explore the potential of another notion for a precise characterization of $\dcfl$ as well as $\cfl$.
We say that a $k$-sna is \emph{(frozen) blank sensitive} if its input-tape head must make a stationary move while reading the frozen blank symbol $B$ on a storage tape on any computation path. Notice that any depth-susceptible $k$-sna is (frozen) blank sensitive. For clarity, the notation $\mathrm{2SDA}_{FBS}$ is used for the collection of all languages recognized by (frozen) blank-sensitive 2-sda's.

\begin{proposition}\label{new-result-FBS}
$2\mathrm{SDA}_{FBS} = \dcfl$.
\end{proposition}

To prove this proposition, we require the key lemma, Lemma \ref{blank-sensitive-to-weak}.
As a slightly weaker notion of  depth-susceptibility, here we introduce another notion of weak depth-susceptible.
A $k$-sna $M$ is \emph{weakly depth-susceptible} if, while a storage-tape head is scanning symbols  $a$ in $\Gamma^{(k)}$, an input-tape head cannot move and, if $a=B$, then an inner state cannot be altered along each computation path of $M$.
As an immediate consequence of weak depth-susceptibility, once a $k$-sna enters a blank region, either it makes a turn or it continues moving deterministically in one direction (as long as an associated computation path terminates).

\begin{lemma}\label{blank-sensitive-to-weak}
Let $k\geq2$. For any (frozen) blank-sensitive $k$-sna $M$, there exists another $k$-sda $N$ such that $L(M)=L(N)$, $N$ is weakly depth-susceptible, and $N$ makes no frozen blank turn. The same statement holds also for a $k$-sna.
\end{lemma}

Meanwhile, we postpone the proof of Lemma \ref{blank-sensitive-to-weak} and focus on  the proof of Proposition \ref{new-result-FBS} with the use of the lemma.
We assume the reader's familiarity with \emph{one-way deterministic pushdown automata} (or 1dpda's).

\begin{proofof}{Proposition \ref{new-result-FBS}}
To show that $\dcfl\subseteq 2\mathrm{SDA}_{FBS}$, we wish to simulate every 1dpda $M=(Q,\Sigma,\{\rhd,\lhd\},\Gamma, \delta,q_0,\bot, Q_{acc},Q_{rej})$ by an appropriately chosen (frozen) blank-sensitive 2-sda $N$. Let $\sigma\in\check{\Sigma}_{\varepsilon}$  ($=\Sigma\cup\{\rhd,\lhd,\varepsilon\}$) and $a\in\Gamma$.
If $M$ replaces $a$ in a stack by a string $\alpha$ while reading $\sigma$, then $N$ first moves its storage-tape head leftward to the first encountered non-$B$ tape cell (which contains $a$), changes $a$ to $B$, moves the tape head rightward to the first encountered $\Box$, writes $\alpha$ over $\Box$s, and makes the tape head stay still. In contrast, if $M$ pops $a$ from the stack while reading $\sigma$, then $N$ moves its storage-tape head leftward  until it hits a non-$B$ tape cell (which contains $a$), changes the content $a$ to $B$, and stays at the current tape cell. This makes $N$ simulate $M$ correctly.

To show that $2\mathrm{SDA}_{FBS}\subseteq \dcfl$, we wish to claim that a depth-2 storage tape is nothing more than a stack.  Let $N=(Q,\Sigma,\{\rhd,\lhd\},\{\Gamma^{(0)},\Gamma^{(1)},\Gamma^{(2)}\},   \delta,q_0, Q_{acc},Q_{rej})$ be any (frozen) blank-sensitive 2-sda. By Lemma \ref{blank-sensitive-to-weak}, it is possible to further assume that $N$ is weakly depth-susceptible and makes no frozen blank turn. Let $a\in\Gamma$ ($=\bigcup_{e\in[0,2]_{\integer}}\Gamma^{(e)}$). We construct an appropriate  1dpda $M$, which can simulate $N$ as follows.
If $N$ writes  a symbol $b$ over $\Box$, then $M$ pushes $b$ to its stack. If $N$ reads $a\in\Gamma^{(1)}$ and changes it to $B$, then $M$ pops $a$ from the stack. While $N$'s storage-tape head stays in a blank region, $N$ does nothing. This simulation is possible because of the weak depth-susceptibility. By the construction of $M$,
it correctly simulates $N$.
\end{proofof}

As a nondeterministic analogue of the above lemma, we immediately obtain the following corollary. By replacing 2-sda's with 2-sna's, we obtain $\mathrm{2SNA}$ (resp., $\mathrm{2SNA}_{FBS}$) from $\mathrm{2SDA}$ (resp., $\mathrm{2SDA}_{FBS}$).

For the storage tape, we use the term \emph{$w$-region} to indicate the tape area where a string $w$ is written as long as the exact location of $w$ is clear from the context. A \emph{left fringe} (resp., a \emph{right fringe}) of the $w$-region refers to the tape cell located on the left side (resp., the right side) of the tape cell holding the leftmost symbol (resp., the rightmost symbol) of $w$. A \emph{blank region} means the consecutive tape cells containing the frozen blank symbol $B$. We implicitly assume that both fringes of a blank region cannot be frozen blank.

\begin{corollary}
$2\mathrm{SNA}_{FBS} = \cfl$.
\end{corollary}


\begin{figure}[t]
\centering
\includegraphics*[height=3.8cm]{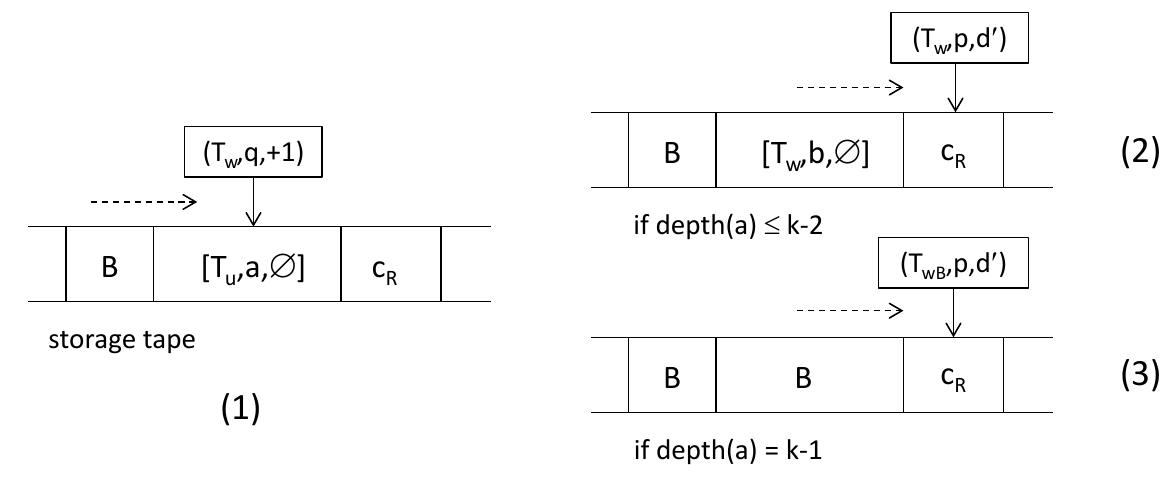}
\caption{A construction of $N$. An input tape is omitted in this illustration. The current tape symbol of a $k$-sda $N$ is $[T_u,a,\setempty]$ and $N$'s current inner state is $(T_w,q,+1)$ in (1). Two possible transitions of $N$ at the next step, depending on the value of $depth(a)$, are depicted in (2) and (3).}\label{fig:blank-sensitive}
\end{figure}



\begin{figure}[t]
\centering
\includegraphics*[height=3.8cm]{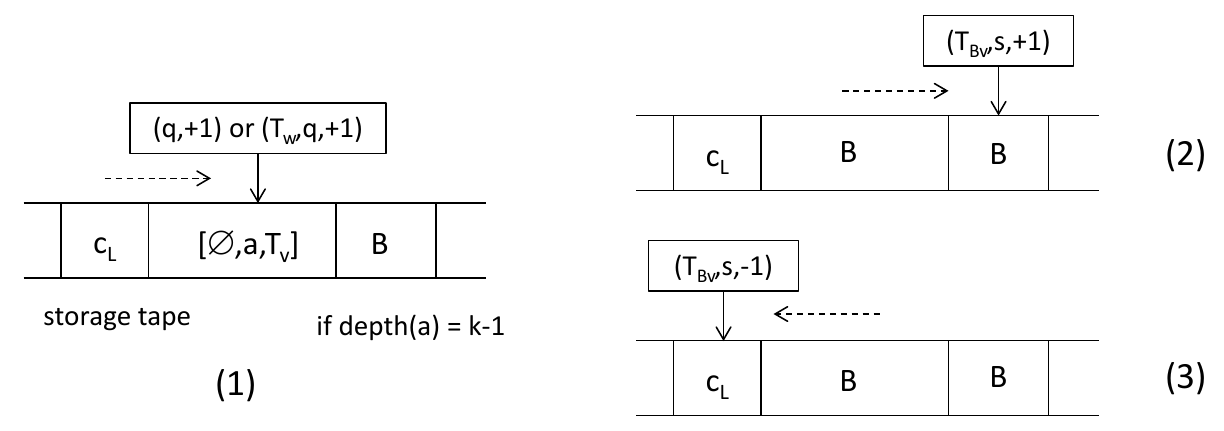}
\caption{Another construction of $N$. The $k$-sda $N$ is scanning $[\setempty,a,T_v]$ with inner state $(q,+1)$ or $(T_w,q,+1)$ in (1). Two possible moves of $N$ are depicted in (2) and (3).}\label{fig:blank-sensitive-second}
\end{figure}


Now, we return to the proof of Lemma \ref{blank-sensitive-to-weak}.
In the proof, a succinct notation $\Gamma^{(\leq t)}$ is meant for the union $\bigcup_{e\in[0,t]_{\integer}}\Gamma^{(e)}$, where $t\in[k]$.

\begin{proofof}{Lemma \ref{blank-sensitive-to-weak}}
This proof is inspired by the proof of the so-called \emph{blank skipping property} of $k$-limited automata \cite{Yam19}. Let $M=(Q,\Sigma,\{\rhd,\lhd\},\{\Gamma^{(e)}\}_{e\in[0,k]_{\integer}}, \delta,q_0,Q_{acc},Q_{rej})$ be any $k$-sda, which is (frozen) blank sensitive. Let $D_1=\{0,+1\}$ and $D_2=\{-1,0,+1\}$ be the set of tape head directions. Note that $\delta$ maps $(Q-Q_{halt})\times \check{\Sigma} \times \Gamma$ to $Q\times \Gamma\times D_1\times D_2$.
To simplify our proof, we introduce another ``equivalent'' $k$-sda $\tilde{M}$ by setting $\tilde{Q}=Q\times D_2$, $\tilde{q}_0=(q_0,+1)$, $\tilde{Q}_{acc}=Q_{acc}\times D_2$, $\tilde{Q}_{rej}=Q_{rej}\times D_2$, and $\tilde{\delta}((q,d),\sigma,a) = \{((p,d''),b,d',d'')\mid (p,b,d',d'')\in\delta(q,\sigma,a)\}$. It follows that $L(\tilde{M})$ coincides with $L(M)$.

Since $M$'s input-tape head stays still on reading $B$, it is possible to simulate the entire behavior of $M$'s storage-tape head with inner states while staying in a blank region. Thus, we can precompute the head direction and the inner state of $M$ when it leaves this blank region. We formalize this argument as follows.

For the construction of the desired $k$-sda $N$, we first introduce three notations $S_{\delta}$, $T_{w}$, and $D_{\delta}$. Given a string $w\in\Gamma^*$ and two pairs $(q,d),(q',d')\in Q\times D_2$, we set $S_{\delta}(q,d\mmid w\mmid q',d') =1$ if,  on a certain computation path of $M$, $M$ enters the $w$-region in direction $d$ with inner state $q$, stays in this region, and eventually leaves the region in direction $d'$ with inner state $q'$, provided that the input-tape head does not move during this process. Otherwise, we set $S_{\delta}(q,d\mmid w\mmid q',d')=0$.
Letting $\ell=|Q\times D_2|$, we next define an $\ell\times\ell$ $\{0,1\}$-matrix $T_w=(t_{ij})_{i,j\in Q\times D_2}$ by setting $t_{ij}=S_{\delta}(i\mmid w\mmid j)$ for any two indices $i,j\in Q\times D_2$.
Over all possible strings $w$, the total number of such matrices $T_w$ is upper-bounded by $2^{\ell^2}$. Finally, for strings $u,v\in\Gamma^*$, $D_{\delta}(T_u\mmid q,d\mmid T_v)$ denotes the set of all pairs $(p,d')\in Q\times D_2$ such that if $d=+1$ (resp., $d=-1$) then
$M$ overwrites $B$ and enters the $v$-region (resp., the $u$-region) on a storage tape in direction $+1$ (resp., $-1$), stays in the $uv$-region, and eventually leaves this $uv$-region in direction $d'$ with inner state $p$.
On contrast, we define $D_{\delta}(\setempty\mmid q,+1\mmid T_v)$ (resp., $D_{\delta}(T_u\mmid q,-1\mmid \setempty)$) to be the set of all pairs $(p,d')$ for which $M$ overwrites a non-$B$ symbol, enters the $v$-region from the left (resp., the $u$-region from the right), stays in this region, and eventually leaves this region in direction $d'$ with inner state $p$.

In what follows, we intend to construct from $\tilde{M}$ the $k$-sda $N$ of the form $(\hat{Q},\Sigma,\{\rhd,\lhd\}, \{\hat{\Gamma}^{(e)}\}_{e\in[0,k]_{\integer}}, \hat{\delta}, \hat{q}_0, \hat{Q}_{acc}, \hat{Q}_{rej})$.
In our construction of $N$, we simulate $\tilde{M}$'s moves almost step by step. During the simulation, we translate a symbol $a\in\Gamma^{(\leq k-1)}$ of $\tilde{M}$ into one of $N$'s symbols $a$, $[T_u,a,\setempty]$, $[\setempty,a,T_v]$, and $[T_u,a,T_v]$. The last three symbols respectively indicate the left fringe and the right fringe of a certain blank region,
and two opposite fringes of two blank regions. Their depths are the same as $a$'s. We also translate an inner state $(q,d)$ of $\tilde{M}$ into either $(q,d)$ or $(T_w,q,d)$. The last of them indicates that the tape cell visited at the previous step is frozen blank. When a tape head enters a blank region $D$, since it eventually leaves $D$, we can check if it moves out of $D$ from the left end or the right end of $D$ without reading any other portion of a given input string.

Let us describe the behavior of $N$ on each storage-tape cell.
In the definition of $N$, for simplicity, we allow a storage-tape head to modify the content of a tape cell even though it makes an intrinsic stationary move. We remark that this relaxation of the stationary requirement of $N$ does not affect the computational power of $N$ because we can remember the tape cell content (using inner states) and overwrite the last modified tape symbol when the tape head moves away to a neighboring tape cell.

\ms

(I)
Consider the case where $N$'s storage-tape head moves from the left to the storage-tape cell currently scanned by $M$ on $x$. Let $c_{L}$ (resp., $c_R$) denote the content of the storage-tape cell located on the left side (resp., the right side) of this tape cell. Assume that $N$ is reading symbol $\sigma$ on its input tape and symbol $\tilde{a}$ (which corresponds to $a\in\Gamma$) on the storage tape in inner state $\tilde{q}$.
We further assume that, on reading $\sigma$, $\tilde{M}$ changes inner state $(q,d)$ to $(p,d')$ with $d=+1$ and $d'\neq0$, and it writes $b$ over $a$.

(1)  Let us consider the case where $\tilde{a}$ equals $a\in\Gamma$. This indicates $c_R\neq B$.

[1] We first focus on the case of $c_L\neq B$, which implies that $\tilde{q}$ is of the form $(q,d)$.

(a) Consider the case where either ($depth(a)\leq k-2$ and $d'=+1$) or ($depth(a)\leq k-3$ and $d'=-1$). Notice that $d'=-1$ means $\tilde{M}$'s making a left turn.  We then force $N$ to write $b$ over $a$ and enter inner state $(p,d')$ from $\tilde{q}$.

(b) Consider the case where either ($k-2\leq depth(a)\leq k-1$ and $d'=-1$) or ($depth(a)=k-1$ and $d'=+1$). In this case, $N$ writes $B$ over $a$ and moves its storage-tape head in direction $d'$ by entering inner state $(T_B,p,d')$ from $\tilde{q}$.

[2] Contrary to [1], we examine the case of $c_L=B$. In this case, $\tilde{q}$ must have the form $(T_w,q,d)$.

(a) In the case where $depth(a)\leq k-2$ and $d'=+1$, since $c_L=B$, $N$ writes $[T_w,b,\setempty]$ over $a$ and enters $(p,d')$ from $\tilde{q}$ by moving in direction $d'$.
In the case where $depth(a)\leq k-3$ and $d'=-1$, in sharp contrast, $N$ writes $[T_w,b,\setempty]$ over $a$. Since $c_L=B$, $\tilde{M}$ enters a blank region. As noted before, it is possible to simulate the entire behavior of $\tilde{M}$ while it stays in this region.
Now, let $(s,d'')$ denote the outcome of $D_{\delta}(T_w\mmid p,-1\mmid \setempty)$. If $d''=+1$, then $\tilde{M}$ eventually leaves the blank region and comes back to the current cell. Thus, we force
$N$ to enter $(T_w,s,+1)$ and to make a storage stationary move.
On the contrary, when $d''=-1$, since $\tilde{M}$ eventually moves off the left end of the blank region, we force $N$ to change its inner state $\tilde{q}$ to $(T_w,s,-1)$ and to make a left turn. The construction (III) ensures that $N$ eventually moves off the left end of the blank region as $\tilde{M}$ does.

(b) Consider the case where $k-2\leq depth(a)\leq k-1$ and $d'=-1$. The tape symbol must be changed to $B$. We write the outcome of  $D_{\delta}(T_w\mmid p,-1\mmid T_B)$ as $(s,d'')$. If $d''=+1$, then $N$ enters inner state $(T_{wB},s,+1)$ by moving to the right. If $d''=-1$, on the contrary, $N$ enters $(T_{wB},s,-1)$ by making a left turn.
Next, let us consider the case where $depth(a)=k-1$ and $d'=+1$. In this case, $N$ writes $B$ over $a$ and move to the right by entering inner state $(T_{wB},p,+1)$.


(2) Here, we examine the case where $\tilde{a}$ has the form $[\setempty,a,T_v]$. This indicates that the current tape cell is the left fringe of a certain blank region, and thus $c_R=B$ follows.

[1] Assume that $c_L\neq B$. This implies that $\tilde{q}$ has the form $(q,d)$.

(a) We first consider the case where $depth(a)\leq k-2$ and $d'=+1$. Since $\tilde{a}=[\setempty,a,T_v]$, $N$ writes $[\setempty,b,T_v]$ over $\tilde{a}$. Let $(s,d'')$ denote $D_{\delta}(\setempty\mmid p,+1\mmid T_v)$. We then force $N$ to enter $(T_v,s,d'')$ from $\tilde{q}$. For the tape head direction, if $d''=+1$, then $N$ moves to the right; however, if $d''=-1$, then $N$ makes a storage stationary move.
In the next case where $depth(a)\leq k-3$ and $d'=-1$, nevertheless, we force $N$ to write $[\setempty,b,T_v]$ over $[\setempty,a,T_v]$ and to make a left turn by changing $\tilde{q}$ to $(p,-1)$.

(b) We next consider the case where $k-2\leq depth(a)\leq k-1$ and $d'=-1$. The machine $N$ writes $B$ over $\tilde{a}$, enters $(T_{Bv},p,-1)$ from $\tilde{q}$, and makes a left turn. The tape cell holding $c_L$ becomes a new left fringe.
In the case where $depth(a)=k-1$ and $d'=+1$, $N$ also writes $B$ over $\tilde{a}$ and changes $\tilde{q}$ to $(T_{wB},p,+1)$.

[2] Opposite to [1], we assume that $c_L=B$, implying that $\tilde{q}$ has the form $(T_w,q,d)$.

(a) Consider the case where $depth(a)\leq k-2$ and $d'=+1$.
We make $N$ write $[T_w,b,T_v]$ over $[\setempty,a,T_v]$. Letting $(s,d'')$ denote $D_{\delta}(\setempty\mmid p,+1\mmid T_v)$, we force $N$ to enter inner state $(T_v,s,d'')$ from $\tilde{q}$.
If $d''=+1$, then $N$ moves to the right. In contrast, if $d''=-1$, then $N$ makes a storage stationary move because the current tape cell is a left fringe of a blank region and $N$'s storage-tape head returns to this fringe after entering this blank region.
In the case where $depth(a)\leq k-3$ and $d'=-1$, $N$ writes $[T_w,b,T_v]$ over $\tilde{a}$. If $(s,d'')$ expresses the outcome of $D_{\delta}(T_w\mmid p,-1\mmid \setempty)$, then $N$ enters inner state $(T_w,s,d'')$. Moreover, when  $d''=+1$, $N$ makes a storage stationary move and, when $d''=-1$, by contrast, $N$ makes a left turn.

(b) Consider the case where either ($k-2\leq depth(a)\leq k-1$ and $d'=-1$) or  ($depth(a)=k-1$ and $d'=+1$). The machine $N$ writes $B$ over $[\setempty,a,T_v]$. We force $N$ to enter $(T_{wBv},s,d'')$ and to move in direction $d''$, where $(s,d'')$ is $D_{\delta}(T_w\mmid p,+1\mmid T_v)$.


(3) Next, we consider the case where $\tilde{a}$ has the form $[T_u,a,\setempty]$. It thus follows that $c_L$ equals $B$, $\tilde{q}$ has the form $(T_w,q,d)$, and $c_R$ cannot be $B$.

(a) Consider the case where $depth(a)\leq k-2$ and $d'=+1$. We force $N$ to  change $\tilde{q}$ to $(p,d')$ and to write $[T_w,b,\setempty]$ over $[T_u,a,\setempty]$ by updating $T_u$ to $T_w$. This update is necessary because a blank region containing $c_L$ might have been modified earlier, and thus $T_w$ may not correctly ``represent'' the blank region.
Next, we consider the case where $depth(a)\leq k-3$ and $d'=-1$. We force $N$ to write $[T_w,b,\setempty]$ over $[T_u,a,\setempty]$. We then calculate $(s,d'')= D_{\delta}(T_w\mmid p,-1\mmid \setempty)$. If $d''=+1$, then $N$ enters inner state $(T_w,s,+1)$ and makes a storage stationary move. By contrast, if $d''=-1$, then $N$ enters $(T_w,s,-1)$ and makes a left turn.

(b) In the case where $k-2\leq depth(a)\leq k-1$ and $d'=-1$, $N$ writes $B$ over $[T_u,a,\setempty]$. Now, we set $(s,d'')$ to be $D_{\delta}(T_w\mmid p,-1\mmid T_B)$. We then change $N$'s inner state from $\tilde{q}$ to $(T_{wB},s,d'')$ and make $N$ move in direction $d''$.
In contrast, when $depth(a)=k-1$ and $d'=+1$, since $N$ does not make a turn, $N$ writes $B$ over $\tilde{a}$ and changes its inner state $\tilde{q}$ to $(T_{wB},p,+1)$.


(4) Next, we examine the case where $\tilde{a}$ has the form $[T_u,a,T_v]$. From this, we conclude that $c_L=c_R=B$, and thus $\tilde{q}$ is of the form $(T_w,q,d)$.

(a) Consider the case where $depth(a)\leq k-2$ and $d'=+1$.  In this case, $N$ writes $[T_w,b,T_v]$ over $[T_u,a,T_v]$ by updating $T_u$ to $T_w$.
Letting $(s,d'')$ be $D_{\delta}(\setempty\mmid p,+1\mmid T_v)$, if $d''=+1$, then $N$ changes $\tilde{q}$ to $(T_v,s,+1)$ and move to the right. On the contrary, when $d''=-1$, $N$ enters $(T_v,s,-1)$ from $\tilde{q}$ and makes a storage stationary move.

(b) Consider the case where $depth(a)\leq k-3$ and $d'=-1$.  We then force $N$ to write $[T_w,b,T_v]$ over $[T_u,a,T_v]$ by updating $T_u$ to $T_w$. Assume that $D_{\delta}(T_w\mmid p,-1\mmid \setempty)=(s,d'')$. If $d''=+1$, then $N$ changes $\tilde{q}$ to $(T_w,s,+1)$ and makes a storage stationary move. When $d''=-1$, however, $N$ enters $(T_w,s,-1)$ and makes a left turn.

(c) Next, we consider the case where either ($k-2\leq depth(a)\leq k-1$ and $d'=-1$) or ($depth(a)=k-1$ and $d'=+1$). Assume that $D_{\delta}(T_w\mmid p,d'\mmid T_v)=(s,d'')$. In this case, $N$ writes $B$ over $[T_u,a,T_v]$ and moves in direction $d''\in\{+1,-1\}$ by entering inner state $(T_{wBv},s,d'')$.

\s

(II) In the second case where $\tilde{M}$ changes inner state $(q,d)$ to $(p,d')$ with $d=-1$ and $d'\neq0$, and it writes $b$ over $a$ on the storage tape, we apply a similar  transformation explained in (I) but taking an opposite head direction.

(III) In the case where $d\neq0$ and $N$ is currently scanning the frozen blank symbol $B$, $N$ keeps the current inner state as well as the storage-tape head direction $d$ taken at the previous step. Since $d\neq0$, this movement of $N$ prohibits $N$ from making the frozen blank turn and from changing its inner state while reading $B$.

(IV) Let us consider the case where $\tilde{M}$ makes a stationary move on reading non-$B$ symbol $a$. In this case, $N$ also makes a stationary move. If $\tilde{q}$ is of the form $(q,d)$ or $(T_w,q,d)$, $N$ changes $\tilde{q}$ to $(p,0)$ or $(T_w,p,0)$, respectively.

\ms

The above construction of $N$ shows that $N$ makes no frozen blank turn and that $N$ is weak depth-susceptible. Thus, the proof is completed.
\end{proofof}


Lemma \ref{blank-sensitive-to-weak} deals with $k$-sda's.
Here, we intend to expand the scope of the lemma to aux-$k$-sna's as follows. An aux-$k$-sna is said to be \emph{(frozen) blank-sensitive} if, while reading $B$ on a storage tape, its input-tape and auxiliary-tape heads make only stationary moves (but inner states may change).
Moreover, an aux-$k$-sna is \emph{weakly depth-susceptible} if its input-tape and auxiliary-tape heads are stationary while scanning any symbol $a$ in $\Gamma^{(k)}$ and, whenever $a=B$, the aux-$k$-sna cannot
alter its inner state.

\begin{lemma}\label{relax-FBS}
Given an aux-$k$-sna $M$ running in polynomial time using log space, if $M$ is (frozen) blank-sensitive, then there exists another aux-$k$-sna $N$ running in polynomial time using log space such that $L(M)=L(N)$, $N$ is weakly depth-susceptible, and $N$ makes no frozen blank turn.
\end{lemma}

\begin{yproof}
In the proof of Lemma \ref{blank-sensitive-to-weak}, we transform a (frozen) blank-sensitive $k$-sna into the desired $k$-sna stated in the lemma.
Since the auxiliary-tape of weakly depth-susceptible aux-$k$-sna $M$ does not move while scanning $B$, we can apply to $M$ the same transformation of the $k$-sna.
\end{yproof}

\subsection{Complexity of LOG$k$SNA}\label{sec:LOGkSNA-aux-sna}

We wish to formally introduce the notation of $\mathrm{LOG}k\mathrm{SNA}$ for any integer $k\geq1$.
For this purpose, we first refer to \emph{log-space, polynomial-time computable functions}, which are functions computed by deterministic Turing machines (or DTMs) equipped with read-only input tapes, rewritable work tapes, and write-once\footnote{An output tape is \emph{write once} if its tape head never moves to the left and, whenever it writes a non-empty symbol, it must move to the adjacent blank cell.} output tapes running in polynomial time using only $O(\log{n})$ work tape space. For convenience, we use the notation $\fl$ to denote the collection of all such functions.
Given two languages $L_1$ and $L_2$, $L_1$ is said to be \emph{$\dl$-m-reducible to} $L_2$, denoted $L_1\Lmreduces L_2$, if there exists a function $f$ (called a \emph{reduction function}) in $\fl$ satisfying that, for any input $x$, $x\in L_1$ is equivalent to $f(x)\in L_2$.

\begin{definition}
Given a family $\CC$ of languages, the notation $\mathrm{LOG}(\CC)$ indicates the collection of all languages $L$ such that there are languages $K$ in $\CC$  for which $L$ is $\dl$-m-reducible to $K$. We often abbreviate $\mathrm{LOG}(\CC)$ as $\mathrm{LOG}\CC$.
\end{definition}

In the deterministic case, Yamakami \cite{Yam21} studied $\mathrm{LOG}k\mathrm{SDA}$. Here, we are particularly interested in  $\mathrm{LOG}k\mathrm{SNA}$.
It was shown in \cite{Yam21} that $\logdcfl\subseteq \logksda{k}\subseteq \logksda{(k+1)} \subseteq \mathrm{SC}^{k}$ for any integer $k\in\nat^{+}$.
Let $\mathrm{NSC}^k$ denote the nondeterministic version of $\mathrm{SC}^{k}$.

\begin{proposition}\label{LOGkSNA-NSC}
For any $k\in\nat^{+}$, $\logcfl\subseteq \logksna{k} \subseteq \logksna{(k+1)} \subseteq \mathrm{NSC}^k$.
\end{proposition}

\section{Useful Properties of aux-$k$-sna's}\label{sec:useful}

As a preparation to establishing a close relationship between aux-$k$-sna's and cascading circuits in Section \ref{sec:circuit-family}, here we discuss useful properties associated with computations of aux-$k$-sna's. In what follows, let
$M = (Q,\Sigma,\{{\rhd},{\lhd}\}, \{\Gamma^{(e)}\}_{e\in[0,k]_{\integer}}, \Theta,\delta, q_0,Q_{acc},Q_{rej})$ denote an arbitrary aux-$k$-sna together with its runtime bound $t(n)$ and work space bound $s(n)$ and consider its computation graph generated on an  input $x$ of length $n$. We further assume that $M$ makes no frozen blank turn on all inputs.
For simplicity, we also assume, without loss of generality, that (1) $M$ must read all input symbols, including the endmarkers (at least once), and (2) before halting, the storage-tape head moves back to $\rhd$.

In our intended simulations between aux-$k$-sna's and cascading circuits explained in Section \ref{sec:circuit-family}, unfortunately, we cannot utilize the notion of ``configurations'' because their (encoded) size is too large to store in $O(\log{n})$ space-bounded auxiliary (work) tapes.
For this very reason, we introduce another tool that partially describes  ``configurations'', which are conventionally called \emph{surface configurations}. Since a computation path is a series of configurations, it also induces a series of corresponding surface configurations.
Formally, a surface configuration $C$ is of the form  $(q,l_1,l_2,w,l_3,a)$, which is directly induced from a configuration $C'=(q,l_1,l_2,w,l_3,u)$ by taking $a=u_{(l_3)}$. This configuration $C'$ is referred to as an \emph{underlying configuration} of $C$.
We remark that each surface configuration completely ignores the content $u$ of the storage tape except for the $l_3$th tape cell. For later convenience, we write $head_3(C)$ instead of $l_3$.
Different from a configuration, such a surface configuration $C$ is  generally obtained from the information on the surface configuration at the previous step and also the previous content of storage-tape cell $l_3$ because a tape head modifies the tape cell before it moves away in a specific direction at the next step. When this tape cell is already frozen blank, nevertheless, the latter information is not necessary since the tape cell is frozen forever.
As special surface configurations, we consider the \emph{initial surface configuration} $P_0$ and \emph{accepting surface configurations} $P_{acc}$ defined as $P_0=(q_0,0,0,\rhd,0,\rhd\Box^{m_2})$ and $P_{acc}=(q_{acc},0,0,\rhd,0,w)$ for any string $w\in \rhd (\Theta-\{\rhd\})^{m_2}$.

Slightly abusing the terminology, we succinctly say that a surface configuration $C$ is a \emph{left turn} (resp., a \emph{right turn}) if $M$'s storage-tape head moves to the left (resp., right) from the current tape cell specified in $C$, knowing that the storage-tape head has come from the left (resp., the right). This last part can be achieved by remembering the storage-tape head direction using inner states as done in the proof of Lemma \ref{blank-sensitive-to-weak}.

Given two surface configurations $P$ and $R$, assuming their underlying configurations $P'$ and $R'$, respectively, we loosely write $P\vdash R$ and $P\vdash_t R$  if $P'$ and $R'$ satisfy $P'\vdash R'$ and $P'\vdash_t R'$. It is important to remember that, for two given surface configurations $P$ and $R$, the relation ``$P\vdash R$'' heavily depends on the choice of their underlying configurations.


For two surface configurations $P=(q,l_1,l_2,w,l_3,a)$ and $P'=(p,l'_1,l'_2,w',l'_3,b)$, if $l_3=l'_3$, then $P$ and $P'$ are said to \emph{have the same depth}. We say that $P$ and $P'$ are
\emph{storage consistent} from $P$ to $P'$, denoted $P\leadsto_{\delta} P'$,  if they have the same depth and there exists a quintuple $(q',c,d_1,d_2,d_3)$ satisfying that  $(q',c,b,d_1,d_2,d_3)\in\delta(q,x_{(l_1)},w_{(l_2)},a)$. Recall that, whenever storage-tape cells are accessed by moving the tape head, their tape symbols must be modified unless they are frozen.
The notion of storage consistency therefore implies that, between $P$ and $P'$, there is no intermediate surface configuration that is storage consistent with both $P$ and $P'$ as long as $depth(a)<k$.


Our goal is to construct a logspace-uniform family $\CC=\{G_n\}_{n\in\nat}$ of $k$-cascading circuits by
transforming the computation graph into another one using specific types of sextuples of surface configurations.
Each circuit $G_n$ is constructed layer by layer inductively from its root to leaves. For this purpose, we first construct a skeleton of the circuit $G_n$ and then modify it to obtain the desired circuit $G_n$.
Meanwhile, we intend to ignore cascading blocks.

In the construction, we assign ``labels'' to constructed gates so that the roles of these gates are clear from their labels. Those labels are called \emph{(surface) configuration duos} and
\emph{(surface) configuration trios}.
A \emph{(surface) configuration duo} has the form $(P,R,t)$ that satisfies the following condition: $P$ and $R$ are surface configurations, they are storage consistent from $P$ to $R$, $t$ ($\in\nat$) denotes the distance between $P$ and $R$ in a certain computation path from $P$ to $R$,  and either (i) $P=R$ and $t=0$ or (ii) $P\neq R$ and $t$ is even.
In particular, assuming that $M$'s storage-tape head returns to the start cell (i.e., cell $0$) and halts, an \emph{accepting duo} means a configuration duo of the form $(P_0,P_{acc},t)$ for the initial surface configuration $P_0$, an accepting surface configuration $P_{acc}$, and an even integer $t\geq1$.
In a similar way, we define \emph{(surface) configuration trios} of the form $(P,Q,R,s,t)$ by requiring that $(P,Q,s)$ and $(Q,R,t)$ are both (surface) configuration duos.

A configuration duo $(P,R,t)$ is called \emph{realizable} if there exist their underlying configurations $P'$ and $R'$ generated by $M$ on $x$, respectively, and there exists a computation path of $M$ on $x$ from $P'$ to $R'$ in  exactly $t$ steps for which the storage-tape cell number of any other configuration on this computation path is larger than $head_3(P')$ and $head_3(R')$. The configuration duos $(C_3,C_{11},8)$ and $(C_{15},C_{17},2)$ obtained from Fig.~\ref{fig:storage-tape-head-move} are realizable but $(C_9,C_{13},4)$ is not.
This realizability relation forms a tree-like structure among configuration duos. Fig.~\ref{fig:computation-tree} illustrates this tree-like structure depicting the realizability relationships among configuration duos induced by the movement given in Fig.~\ref{fig:storage-tape-head-move}.
Such a tree-like structure whose root is an accepting surface configuration correctly represents an accepting path.
In this tree-like structure, two configuration duos $(P_1,R_1,t_1)$ and $(P_2,R_2,t_2)$ are said to be \emph{linked} if $R_1$ and $R_2$ have the same depth and there exists a computation path from $R_1$ to $P_2$. Similarly, we say that (surface) configuration trio $(P_1,Q_1,R_1,s_1,t_1)$ is \emph{linked} to $(P_2,R_2,t_2)$ if $(Q_1,R_1,t_1)$ is linked to $(P_2,R_2,t_2)$. In contrast, $(P_1,R_1,t_1)$ is \emph{linked} to $(P_2,Q_2,R_2,s_2,t_2)$ if $(P_1,R_1,t_1)$ is linked to $(P_2,Q_2,s_2)$.

\begin{lemma}\label{tree-like-realizable}
$M$ accepts $x$ iff there exists a tree-like structure induced by the realizability relation such that its root is an accepting duo $(P_0,P_{acc},\mathrm{Time}_M(|x|))$ and is realizable.
\end{lemma}


\begin{figure}[t]
\centering
\includegraphics*[height=6.6cm]{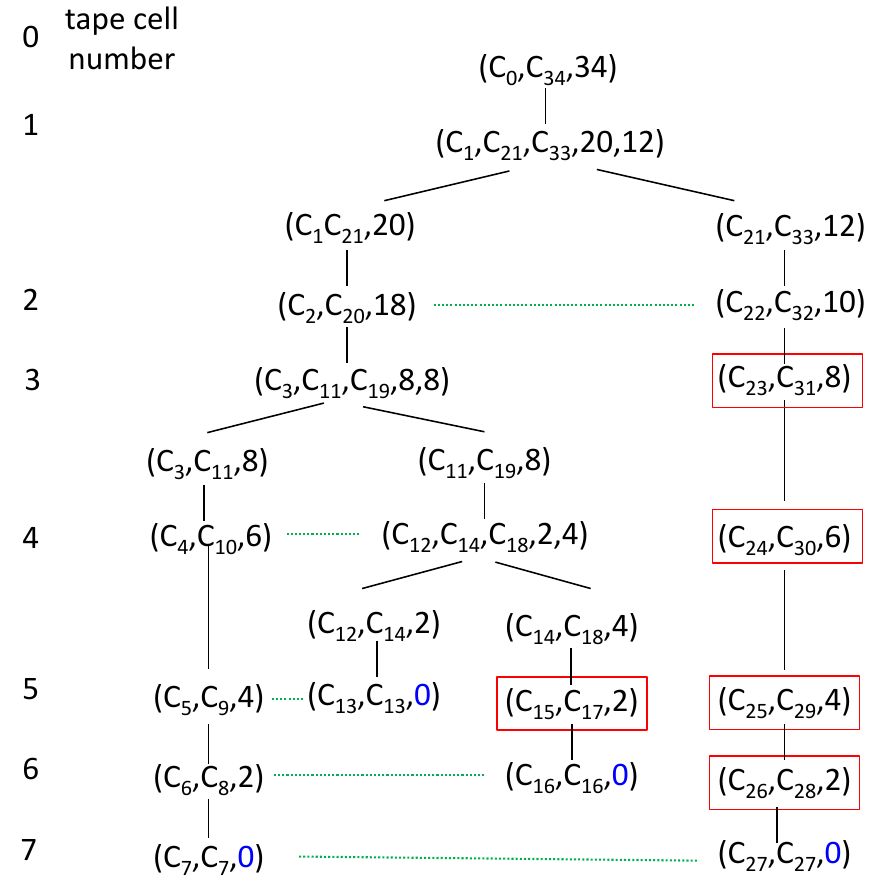}
\caption{A tree-like structure that represents the transitional relationships of all surface configuration duos corresponding to the movement depicted in  Fig.~\ref{fig:storage-tape-head-move}. A configuration duo inside a blue box indicates that its currently-scanned storage-tape cells are all frozen blank. A green dotted line indicates a link between two configuration duos. Here, the node $(C_{12},C_{14},2)$ depends on $(C_{13},C_{13},0)$ and $(C_{12},C_{14},C_{18},2,4)$, which is linked from $(C_{4},C_{10},6)$, and $(C_{16},C_{16},0)$ depends on $(C_{15},C_{17},2)$ and also it is linked from   $(C_6,C_{8},2)$. On the contrary, $(C_{15},C_{17},2)$ depends only on $(C_{14},C_{18},4)$ and $(C_{16},C_{16},0)$ but is not linked from  $(C_{13},C_{13},0)$ because the storage tape cells of $(C_{15},C_{17},2)$ are already frozen blank.}\label{fig:computation-tree}
\end{figure}


A \emph{layered block} of configuration duos is a series $\QQ= \{(P_j,P_{j+1},t_j)\}_{j\in[m]}$ of configuration duos with an odd number $m$ for which $P_i$ and $P_j$ are storage consistent for any pair $i,j\in[m+1]$. 
Such a layered block $\QQ$ is said to be \emph{realizable} if there exists a computation path such that it starts at $P_1$, passes $P_2,P_3,\ldots,P_{m}$ one by one in order in exactly $t_1,t_2,\ldots,t_{m}$ steps, respectively, and finally ends at $P_{m+1}$ and that, for each index $j\in[m]$, $(P_j,P_{j+1},t_j)$ is realizable along this computation path.
As a quick example,  the block of configuration duos $\QQ=\{(C_4,C_{10},6), (C_{1},C_{21},20), (C_{21},C_{33},12)\}$ and $\{(C_3,C_{11},8),(C_{11},C_{19},8)\}$ depicted in  Fig.~\ref{fig:computation-tree} are both realizable.

We expand the scope of the binary relation $\vdash$ to layered blocks of configuration duos in the following way. Let $(P,R,t)$ denote a target configuration duo and let $\QQ= \{(P_j,P_{j+1},t_j)\}_{j\in[m]}$ be any layered block.
We write $\QQ \vdash (P,R,t)$ if (i) when $m=1$, it follows that  $t=t_1+s+s'$, $P\vdash_{s} P_1$, and $P'_2\vdash_{s'} R$  for certain numbers $s,s'\in\nat^{+}$ and (ii) when $m\geq3$, it follows that $t=\sum_{j\in[m]}t_j +s+s'$, $P\vdash_{s} P_1$, and  $P_{m+1}\vdash_{s'} R$ for certain numbers $s,s'\in\nat^{+}$.
Moreover,  we write $\vdash (P,R,i)$ if $P=R$, $i=0$, and $P$ is a left turn.


The following statement is immediate.

\begin{lemma}\label{realizable}
Let $\QQ$ denote any layered block $\{(P_j,P_{j+1},t_{j})\}_{j\in[m]}$. If $\QQ$ is realizable and $\QQ\vdash (P,R,t)$, then $\QQ'= \QQ\cup\{(P,P_1,s),(R_{m+1},R,s')\}$ is also realizable for certain numbers $s,s'\geq1$.
\end{lemma}

\begin{proof}
Let $\QQ=\{(P_j,P_{j+1},t_{j})\}_{j\in[m]}$ and assume that $Q$ is realizable. Thus, there is a computation path from $P_1$ to $P_{m+1}$. Since $\QQ\vdash(P,R,t)$, there are two numbers $t,t'\geq1$ and a computation path  on which $P\vdash_{s} P_1$ and $P_{m+1}\vdash_{s'} R$ with $t=\sum_{j\in[m]}t_j+s+s'$.
We then define $\QQ'= \QQ\cup\{(P,P_1,s),(R_{m+1},R,s')\}$.
Hence, $\QQ'$ is realizable.
\end{proof}

The following basic property holds for aux-$k$-sna's.

\begin{lemma}\label{recursion-property}
Fix a computation path of the aux-$k$-sna $M$ on input $x$. Let $(P,R,t)$ denote any realizable configuration duo of $M$ on $x$ with $t\geq2$ along this computation path. There exists a unique layered block $\QQ= \{(P'_j,P'_{j+1},t_j)\}_{j\in[m]}$ with $m\geq1$ or $\QQ=\setempty$ such that $\QQ$ is realizable along this computation path and $\QQ\vdash (P,R,t)$.
\end{lemma}

\begin{proof}
We fix a given computation path of $M$ and assume that  $P=(q,l_1,l_2,w,l_3,u)$ and $R=(q',l'_1,l'_2,w',l'_3,u')$, where  $head_3(P)=l_3$ and $head_3(R)=l'_3$.
Since $(P,R,t)$ is a configuration duo, $head_3(P)=head_3(R)$ follows.
Assume that $(P,R,t)$ is realizable along this computation path. The uniqueness of $\QQ$ comes from the fact that a computation path is fixed.

We prove by induction on the value of $l_3$ that the lemma's conclusion is true for all $P$ and $R$ with $head_3(P)=head_3(R)=l_3$. Assume by induction hypothesis that the lemma's conclusion is true for the value $l_3$. Let us consider all configurations associated with $l_3-1$. On the computation path  between $P$ and $R$, $M$'s storage-tape head visits cell $l_3-1$ at least twice.

(1) Assume that $P=R$ and $t=0$.  Since $(P,R,t)$ is realizable,  $(P,R,t)$ must be a left turn. We then define $\QQ=\setempty$, and thus we obtain $\vdash(P,R,t)$.

(2)  In what follows, we assume that $P\neq R$ and $t\geq2$. This implies that $(P,R,t)$ is not a turn.

(a) Consider the case where there is no right turn at $l_3-1$ on the computation path between $P$ and $R$. We can choose two surface configurations $P'_1$ and $P'_2$ with $head_3(P'_1)=head_3(P'_2)=l_3-1$ such that $P\vdash P'_1$ and $P'_2\vdash R$ along the given computation path. Let $\QQ=\{(P'_1,P'_2,t-2)\}$. It follows that $\QQ\vdash (P,R,t)$.

(b) Consider the case where there are right turns at $l_3-1$. Choose $P'_1,P'_2,\ldots,P'_{m+1}$ of surface configurations such that $P\vdash P'_1$, $P'_{m+1}\vdash R$, and $P'_j$ is a right turn at $l_3-1$ for any index $j\in[m+1]$. Let $t_j$ denote the number of steps from $P'_j$ to $P'_{j+1}$ along the given computation path. We define $\QQ=\{(P'_j,P'_{j+1},t_j)\}_{j\in[m]}$. This series $\QQ$ is realizable  and we also obtain $\QQ\vdash(P,R,t)$.
\end{proof}


For any aux-$k$-sna making no frozen blank turn, the number of turns at each storage-tape cell is upper-bounded on any computation path of the aux-$k$-sna on every input.

\begin{lemma}\label{bound-sna-turn}
For any aux-$k$-sna, if it makes no frozen blank turn, then it makes less than or equal to $\ceilings{k/2}$  turns at each storage-tape cell on any computation path of $M$ on each input.
\end{lemma}

\begin{yproof}
Let $M$ denote any aux-$k$-sna, running in polynomial time using log space, which makes no frozen blank turn. Note that, for each storage-tape cell $j$, if $M$ makes more than $k/2$ turns at cell $j$ (ignoring just passing through cell $j$), cell $j$ becomes frozen blank. After cell $j$ becomes frozen blank, $M$'s storage-tape head cannot make any turn at cell $j$ because such a turn becomes a frozen blank turn.
\end{yproof}

We strengthen Lemma \ref{bound-sna-turn} for (frozen) blank-sensitive  aux-$k$-sna's. We say that an aux $k$-sna $M$ is \emph{right-turn restricted} if, for any input $x$ and  on any computation path of $M$ on $x$, $M$ makes right turns no more than once at each storage-tape cell.
Fig.~\ref{fig:storage-tape-head-move} depicts a computation path of a right-turn restricted aux-$4$-sna. The following lemma shows a conversion of a (frozen) blank-sensitive aux-$k$-sna into a weakly depth-susceptible, right-turn restricted aux-$k$-sna.

\begin{lemma}\label{once-right-turn}
For any (frozen) blank-sensitive aux-$k$-sna $M$ running in polynomial time and log space, there always exists an aux-$k$-sna $N$ running in polynomial time and log space satisfying that $L(M)=L(N)$, $N$ is weakly depth-susceptible, and $N$ is right-turn restricted.
\end{lemma}

\begin{yproof}
Take an arbitrary (frozen) blank-sensitive aux-$k$-sna $M$ that runs in polynomial time using log space. By
Lemma \ref{relax-FBS}, it is possible to assume that $M$ is weakly depth-susceptible and makes no frozen blank turn.
Firstly, we wish to construct an aux-$k$-sna $N$ for the language $L(M)$ such that $N$ is (frozen) blank-sensitive and right-turn restricted.

Let $\Gamma=\bigcup_{e\in[0,k]_{\integer}} \Gamma^{(e)}$ denote a storage alphabet of $M$. We make a correspondence between cell $i$ of $M$'s storage tape and a block of $k$ cells indexed from $(i-1)k+1$ to $ik$ of $N$'s storage tape. In what follows, whenever $M$'s storage-tape head moves from the left (resp., the right) to cell $i$, we make $N$'s storage-tape head move from the left (resp., the right) to cell $(i-1)k+1$ (resp., $ik-1$). For convenience, we call this block ``block $i$''.
More precisely, a tape symbol $\sigma$ (except for $\rhd$, $\Box$, and $B$) written at cell $i$ of $M$ is expressed as a string of the form $B^j\track{\sigma}{1}^{j}\sigma^{k-j-1}$ written on cells indexed between $(i-1)k+1$ and $ik$, where $j$ is the number of right turns made by $M$ at cell $i$. For this operation, we need to count
the number of $B$s in each block. Since $k$ is a constant, this counting can be carried out using only $N$'s inner states.
We further add extra steps to $N$ in the following way to guarantee the right-turn restriction.

(a) Consider the case where $M$'s storage-tape head comes from the left to cell $i$. In this case, $N$'s storage-tape head also moves from the left and it is now scanning the leftmost cell of block $i$. If $depth(\sigma)<k-1$ and $M$'s tape head changes the current symbol, say, $\sigma$ to $\tau$ and moves to the right, then we also move $N$'s tape head to the right by rewriting $B^j\track{\sigma}{1}\sigma^{k-j-1}$ to $B^j\track{\tau}{1}\tau^{k-j-1}$. In the case of $depth(\sigma)=k-1$, since $M$ changes $\sigma$ to $B$,  $B^j\track{\sigma}{1}\sigma^{k-j-1}$ is changed to $B^k$.
On the contrary, if $depth(\sigma)<k-2$ and $M$'s tape head makes a left turn by changing $\sigma$ to $\tau$, we move $N$'s tape head rightward to the right end of block $i$ by changing $B^j\track{\sigma}{1}\sigma^{k-j-1}$ to $B^j\track{\hat{\sigma}}{1}\hat{\sigma}^{k-j-2}\sigma$, make a left turn at the end of block $i$, and change  $B^j\track{\hat{\sigma}}{1}\hat{\sigma}^{k-j-2}\sigma$ to $B^j\track{\tau}{1}\tau^{k-j-1}$, where $\hat{\sigma}$ is a new symbol associated with $\sigma$ and its depth is $depth(\sigma)+1$. This is possible because $depth(\tau)=depth(\sigma)+2<k$.
If $k-2\leq depth(\sigma)\leq k-1$, then $M$'s tape head, in contrast, changes $\sigma$ to $B$. Correspondingly, $N$ first changes $B^j\track{\sigma}{1}\sigma^{k-j-1}$ to $B^j\track{\hat{\sigma}}{1}\hat{\sigma}^{k-j-2}\sigma$, makes a left turn at the right end of block $i$, and changes it to $B^k$.

(b) Consider the case where $M$'s storage-tape head comes from the right to cell $i$. Note that $N$'s storage-tape head also moves from the right to the right end of block $i$. If $depth(\sigma)\leq k-2$, $M$'s tape head changes $\sigma$ to $\tau$, and it moves to the left, then we move $N$'s tape head leftward by rewriting $B^j\track{\tau}{1}\tau^{k-j-1}$ over $B^j\track{\sigma}{1}\sigma^{k-j-1}$. In contrast, if $depth(\sigma)\leq k-3$ and $M$ makes a right turn, then we  move $N$'s tape head to the left by changing $\sigma$ to $\tau$ one by one until scanning $\track{\sigma}{1}$. The machine $N$ then makes a right turn and rewrites the rest of the block from $\track{\sigma}{1}\sigma^{k-j-1}$ to $B\track{\tau}{1}\tau^{k-j-2}$.
If $depth(\sigma)\geq k-2$ and $M$'s tape head makes a right turn, then $N$  writes $B^k$. Similarly, if $depth(\sigma)=k-1$ and $M$'s tape head makes a left turn, then $N$ overwrites $B^k$.

(c) If $M$'s tape head makes a storage stationary move at cell $i$, then $N$'s tape head does the same.

By the above construction of $N$, we conclude that $N$ is (frozen) blank sensitive and right-turn restricted. Moreover, it makes no frozen blank turn. As in the proof of Lemma \ref{relax-FBS}, we conduct the simulation of $N$ by another aux-$k$-sna, say, $N'$. It thus follows that $N'$ is weakly depth-susceptible and makes no frozen blank turn. Moreover,  examining this simulation shows that $N'$ is also right-turn restricted.
\end{yproof}

\section{Cascading Circuit Families}\label{sec:circuit-family}

The main goal of this work is to give a circuit characterization of languages in  $\auxsna\mathrm{depth,\!space,\!time}(2k,O(\log{n}),n^{O(1)})_{FBS}$. Toward this goal,  we formally introduce  a key concept of ``cascading'' Boolean circuits in Section \ref{sec:cascading}. The precise statement of the characterization is given in Section \ref{sec:simulation}.

\subsection{Semi-Unbounded Fan-in Cascading Circuits}\label{sec:cascading}

A \emph{Boolean circuit} for inputs of size $n$ is an acyclic directed graph, in which all nodes (called \emph{gates}) are labeled with Boolean operators except for indegree-$0$ and outdegree-$0$ nodes, which are respectively called \emph{input gates} and \emph{output gates}. Edges in the circuit are often called \emph{wires}.
The \emph{size} of a circuit is the total number of gates in it and the \emph{depth} is the length of the longest path from an input gate to an output gate. We mostly consider families $\{C_n\}_{n\in\nat}$ of circuits of size polynomial in $n$ and depth logarithmic in $n$. The \emph{fan-in} (resp., \emph{fan-out}) of a gate is the number of incoming wires to (resp., outgoing wires from) the gate.
In general, a gate is said to have \emph{bounded fan-in} if the fan-in of the gate is upper-bounded by a fixed constant independent of $n$. Otherwise, the gate has \emph{unbounded fan-in}. Similarly, we define \emph{(un)bounded fan-out} gates.

As fundamental gates, we use AND ($\wedge$), OR ($\vee$), and negation ($\overline{x}$).
As customary in circuit complexity theory, AND as well as OR gate takes only one output and the negation is applied only to input variables.
We also use a special, bounded fan-in, unbounded fan-out AND gate, distinctively denoted AND$_{(\omega)}$, where the subscript ``$(\omega)$'' indicates the property of ``unbounded'' fan-out, and we reserve the notation AND for an AND gate of fan-out  $1$.
Although it is possible to replace an AND$_{(\omega)}$ gate by a number of AND gates together with  appropriately duplicated subcircuits rooted at this  AND$_{(\omega)}$ gate, the size of the resulting circuit may be considerably larger than the original one. These AND$_{(\omega)}$ gates are used to build cascading (sub)circuits.

In what follows, nevertheless, we use three types of gates:
AND gates of fan-in $2$ and fan-out $1$, AND$_{(\omega)}$ gates of bounded fan-in and unbounded fan-out\footnote{More precisely, we demand that the fan-in of each AND$_{(\omega)}$ is at most a fixed constant, say, $c$ independent of $n$ for $G_n$ and the fan-out is at least $1$. In the special case of AND$_{(\omega)}$ having fan-in $2$ and fan-out $1$, AND$_{(\omega)}$ is the same as AND. Even though, we tend to keep the notation AND$_{(\omega)}$ to distinguish it from AND.}, and OR gates of unbounded fan-in\footnote{For later convenience, we allow each OR gate to take only one input although the use of such an OR gate is redundant.} and fan-out $1$. Since the negations of input bits are integrated into  parts of inputs,  there are no explicit use of NOT gates.
A family of  \emph{semi-unbounded fan-in circuits} refers to a circuit family for which there is a constant $\ell\geq1$ such that any path from the root to a leaf node in any circuit in the family has at most $\ell$ consecutive gates of AND.

It should be remarked that any gate at depth $i$ takes inputs from gates at depth less than $i$. We assume that a circuit is \emph{layered}\footnote{This does not mean that the circuit is ``leveled''; that is, all gates at level $i$ take inputs from only gates at level $i-1$. See, e.g., Fig.~\ref{fig:cascading-block} for an illustration of a layered circuit.} so that
all gates are placed in certain layers and all gates placed in the same layer are indexed from left to right so that we can easily find the right and the left adjacent gates (if any) of any gate in the same layer. See Fig.~\ref{fig:circuit-simulation} for a quick example of layered circuits.


\begin{figure}[t]
\centering
\includegraphics*[height=5.5cm]{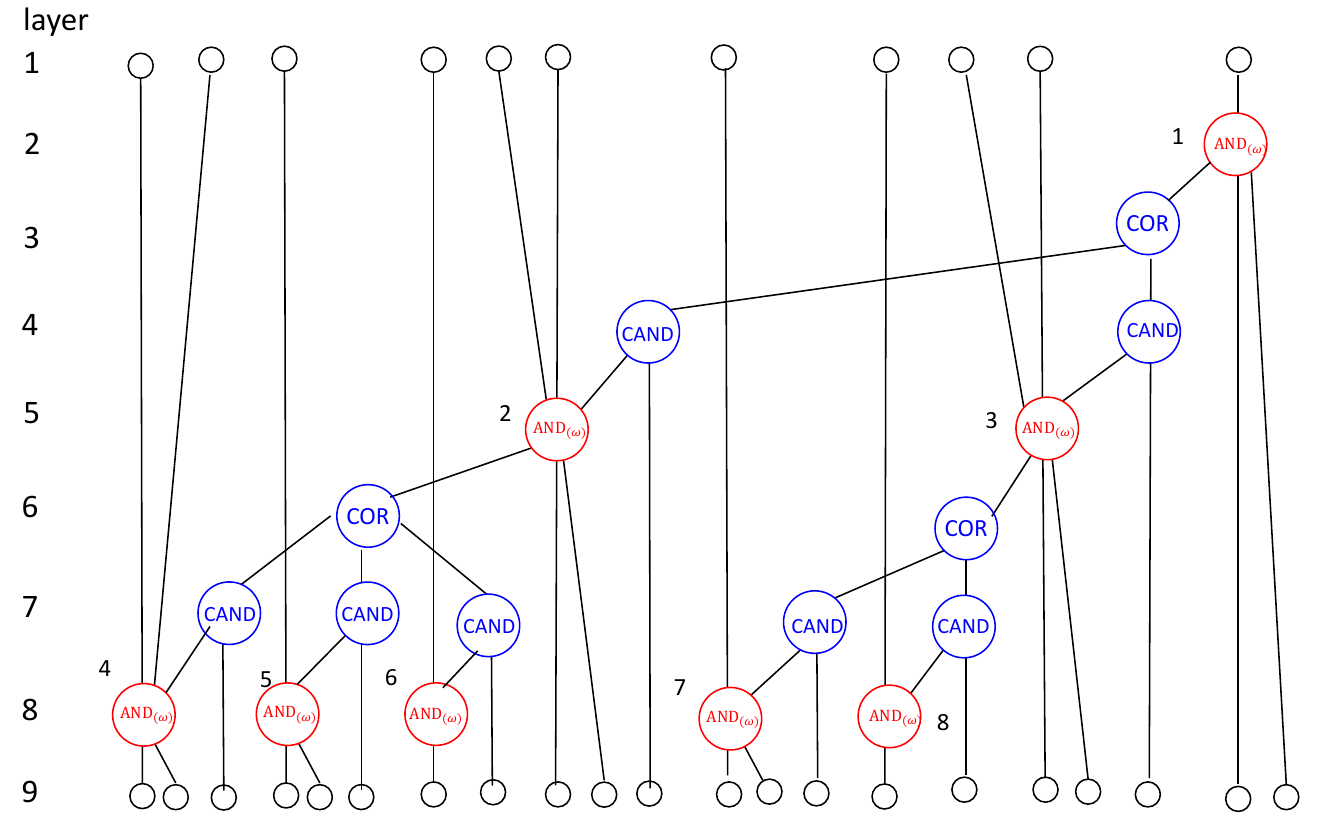}
\caption{A cascading block, which is a subcircuit composed of  AND$_{(\omega)}$, CAND,  and COR gates marked as large circles. The smaller circles at layers $1$ and $9$ indicate output gates and input gates of this subcircuit, respectively. The link length of this subcircuit is $3$. The numbers adjacent to some circles indicate gate labels. For example, gate 4 as well as gates $5$ and $6$ is linked to gate 2, which is further linked to gate 1. In contrast, gates $5$ and $7$ are not linked. Similarly, gates $5$ and $6$ are not linked.}\label{fig:cascading-block}
\end{figure}


\begin{definition}
A \emph{cascading block} $D$ is a special layered subcircuit composed of AND$_{(\omega)}$, AND, and OR gates (the last two of which are distinctively called CAND and COR gates) with numerous input and output gates satisfying the following conditions.  Let $i$ be any natural number.

\renewcommand{\labelitemi}{$\circ$}
\begin{enumerate}\vs{-1}
  \setlength{\topsep}{-2mm}%
  \setlength{\itemsep}{1mm}%
  \setlength{\parskip}{0cm}%

\item[(1)] Gates in the first and the last layers are used as outputs and inputs of this block. These gates should be replaced by appropriate gates when the block is inserted into a larger circuit.

\item[(2)] At layer 2, there is only one AND$_{(\omega)}$ gate and its fan-out is $1$.

\item[(3)] The subcircuit $D$ contains AND$_{(\omega)}$ gates only at layer $3i+2$,  COR gates at layers $3i+3$, and CAND gates at layer $3i+4$.

\item[(4)] Each  AND$_{(\omega)}$ gate at layer $3i+5$ is directly connected  to (possibly) output gates at the first layer and exactly one CAND gate at layer $3i+4$.

\item[(5)] Each  AND$_{(\omega)}$ gate at layer $3i+2$ is directly  connected from (possibly) input gates at the last layer and  exactly  one COR gate at layer $3i+3$.

\item[(6)] Each COR gate at layer $3i+3$ is directly connected only from CAND gates at layer $3i+4$.

\item[(7)] For any two distinct AND$_{(\omega)}$ gates, there is no more than one directed path between them.
\end{enumerate}\vs{-1}
\end{definition}

The \emph{bottom AND$_{(\omega)}$ gates} of $D$ refer to the AND$_{(\omega)}$ gates at the bottom layer and the unique AND$_{(\omega)}$ gate at the top layer is called the \emph{top AND$_{(\omega)}$ gate} of $D$. Fig.~\ref{fig:cascading-block} illustrates an example of cascading block.

We say that two AND$_{(\omega)}$ gates are \emph{linked} if there is a directed path from one of them at a higher layer to the other at a lower layer passing through CAND and COR gates and (possibly) other AND$_{(\omega)}$ gates. Such a path is called a \emph{link} between them.
A link is said to be \emph{direct} if no other AND$_{(\omega)}$ gate lies in the link.
In a special case where an AND$_{(\omega)}$ gate $C$ is linked to no other AND$_{(\omega)}$ gates, $C$ is considered to form a cascading block of a single AND$_{(\omega)}$ gate of fan-in $1$ and this gate is not necessarily connected to COR gates.

The \emph{link length} of a cascading block $D$ is the maximum number of AND$_{(\omega)}$ gates (not including CAND gates) that appear in any path from an input gate to an output gate of $D$.
The cascading block of Fig.~\ref{fig:cascading-block} has link length $3$.
The three AND$_{(\omega)}$ gates at layer 8 from the left (i.e., gates 4, 5,  and 6) are all linked to the leftmost AND$_{(\omega)}$ gate at layer 5 (i.e., gate 2). Similarly, the two AND$_{(\omega)}$ gates at layer $8$ from the right (i.e., gates 7 and 8) are linked to the rightmost AND$_{(\omega)}$ gate at layer $5$ (i.e., gate 3).


\begin{figure}[t]
\centering
\includegraphics*[height=7.0cm]{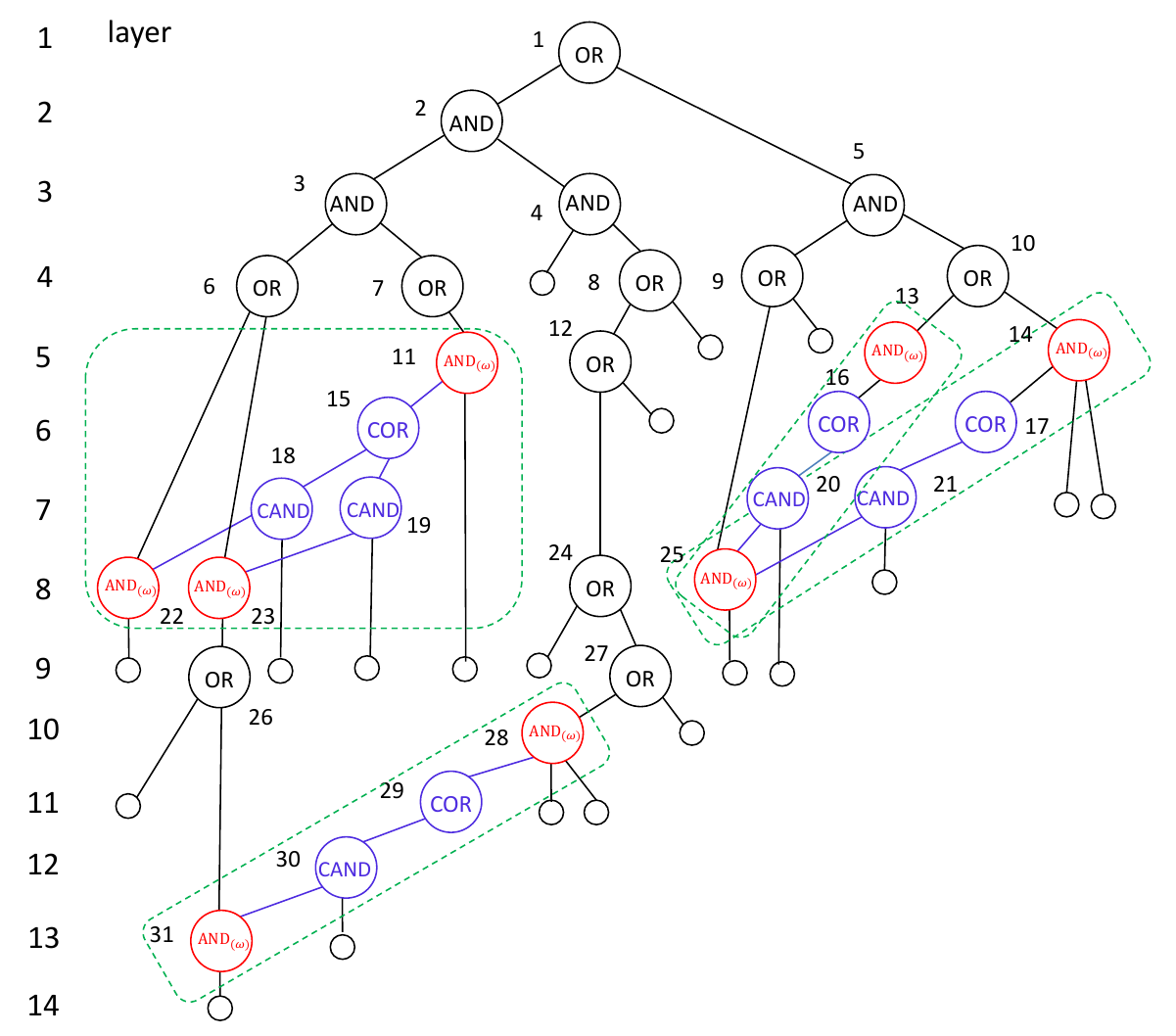}
\caption{A simple $2$-cascading circuit. Small circles are input gates of the circuit. The numbers adjacent to circles are gate labels. Five cascading blocks excluding their top and bottom layers (since the corresponding gates are already replaced by other gates) are marked by green dotted boxes. For the left cascading blocks, gate 3 is the minimum common ancestor of gates 11 and 23, and gate 2 is the minimum common ancestor of gates 28 and 31.
The subcircuit rooted at gate 2 is a cascading semi-circuit. Similarly, the subcircuit rooted at gate 5 is also a cascading semi-circuit composed of two cascading blocks. A basis subcircuit is made up of gates $1$, $2$, and $5$.}\label{fig:circuit-simulation}
\end{figure}


Now, we wish to embed those cascading blocks into a larger circuit, which we intend to call a \emph{cascading circuit}. When two cascading blocks $A$ and $B$ are embedded into such a circuit, $A$ is said to be an \emph{ancestor} of $B$ (or $B$ is a \emph{descendant} of $A$) in this circuit if the following condition holds: for any two AND$_{(\omega)}$ gates $g_1$ and $g_2$ in $A$ and $B$, respectively, if $g_1$ and $g_2$ are connected by a path, then $g_1$ must be an ancestor of $g_2$ in this circuit (which is viewed as a graph).
Those two cascading blocks $A$ and $B$ of cascading lengths $k_{A}$ and $k_{B}$, respectively, are further said to be \emph{orderly} if, for any number $i\in[\min\{k_{A},k_{B}\}]$,  the $i$th AND$_{(\omega)}$ gate $C$ in $A$ is an ancestor of the $i$th AND$_{(\omega)}$ gate in $B$.

A cascading block $D$ in a circuit is \emph{the leftmost} if there is no cascading block $C$ that is located on the left side of $D$, namely, for any two gates $g_C$ and $g_D$ taken from $C$ and $D$, respectively, if they are on the same layer, then $g_C$ is located on the left side of $g_D$.

When we delete all input gates from a given circuit $G$, the resulting circuit is called the \emph{leaf-free version} of $G$. Given a circuit $G$, the \emph{leafless fan-in} of an AND$_{(\omega)}$ gate of $G$ refers to the number of inputs wired directly to this gate except for the inputs wired  directly from input gates. Notice that the leafless fan-in may be zero if all inputs to this gate are wired directly from input gates of $G$.

Given a cascading block $D$, the \emph{cascading length} of $D$ is defined to be one plus the sum of the leafless fan-ins of any AND$_{(\omega)}$ gates of $D$ except for the top AND$_{(\omega)}$ gate. In particular, when $D$'s link length is $1$, the cascading length of $D$ equals $1$. If all AND$_{(\omega)}$ gates of a cascading block have leafless fan-ins $1$, then the cascading length exactly matches the link length.

To simplify a later description of circuits, we delete some portions of the circuits. Given a circuit $G$, for each OR gate (including COR gates) of $G$, we choose exactly one child and delete all subcircuits rooted at the children that are not chosen as well as their associated connecting wires. The resulting subcircuits constitute only of OR gates of fan-in $1$, AND$_{(\omega)}$ gates, and AND gates. We call such subcircuits \emph{decisive fragments} of $G$. An example of such decisive fragments of the subcircuit rooted at gate 2 in Fig.~\ref{fig:circuit-simulation} is depicted as  Fig.~\ref{fig:unlinked-version}(1).
Given a circuit $G$ containing cascading blocks, its \emph{unlinked version} is the circuit obtained directly
from $G$ by deleting all CAND and COR gates and their connecting wires. For the subcircuit rooted at gate 2 in Fig.~\ref{fig:circuit-simulation}, its unlinked version is depicted as Fig.~\ref{fig:unlinked-version}(2).


\begin{figure}[t]
\centering
\includegraphics*[height=6.3cm]{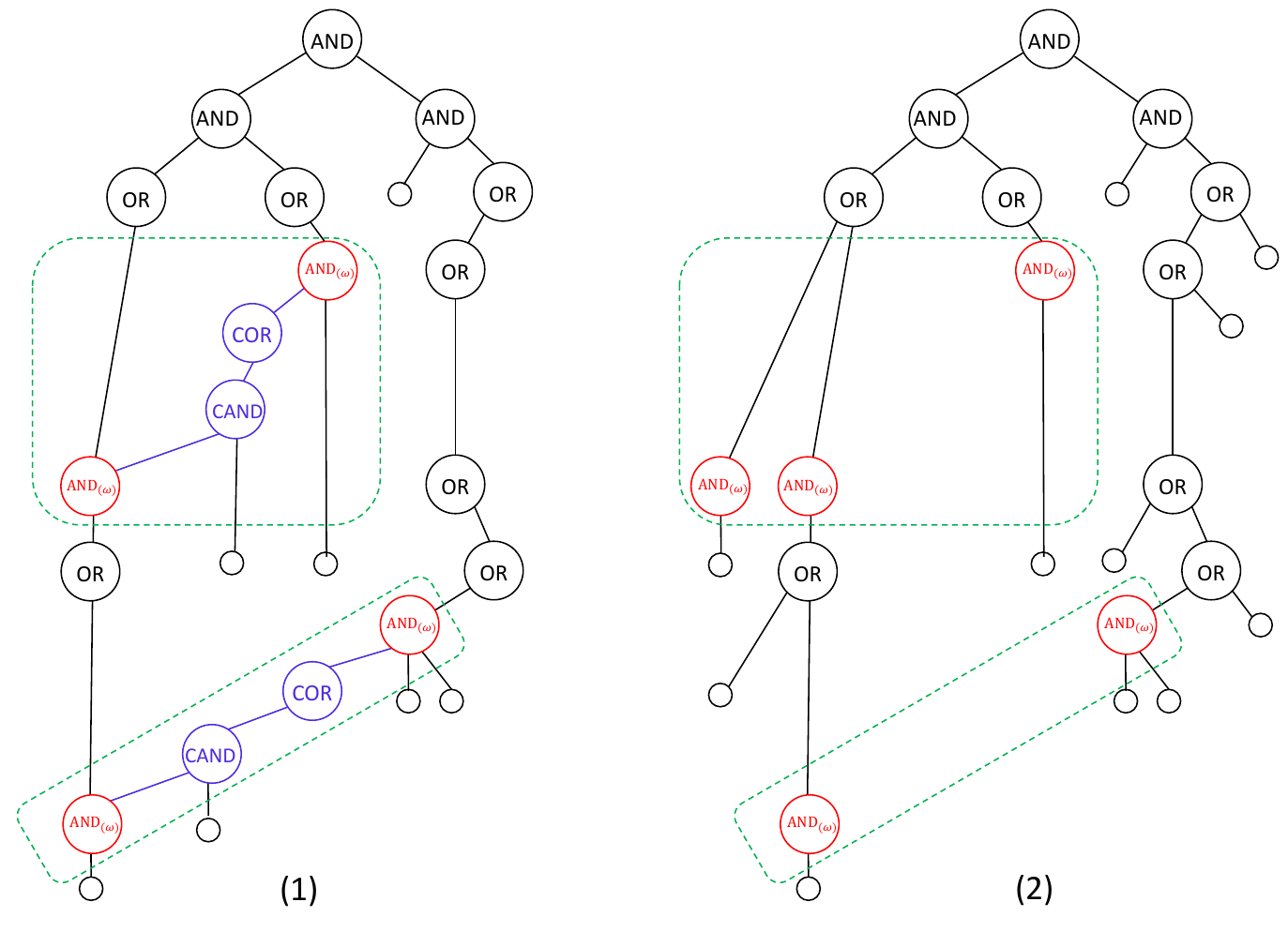}
\caption{(1) A decisive fragment of the subcircuit rooted at gate 2 of Fig.~\ref{fig:circuit-simulation}. All OR gates (including COR gates) have fan-in 1 and fan-out 1. The two cascading blocks marked by the green dotted boxes are orderly.
(2) The unlinked version of the same subcircuit rooted at gate 2 of   Fig.~\ref{fig:circuit-simulation}. This unlinked version forms a tree.}\label{fig:unlinked-version}
\end{figure}


Before giving a formal definition of cascading circuits, we define ``cascading semi-circuits'', which constitute crucial parts of cascading circuits.
When we embed a cascading block into a larger circuit, we refer to this block in the circuit using the \emph{block top layer number} (resp., \emph{block-bottom layer number}), which is the layer number of the top (resp., bottom) AND$_{(\omega)}$ gate of the block in the circuit. In Fig.~\ref{fig:circuit-simulation}, for the cascading block containing three  AND$_{(\omega)}$ gates $11$, $22$, and $23$, its block-top layer number and block-bottom layer number are $5$ and $8$, respectively.

\begin{definition}
A \emph{cascading semi-circuit} is a circuit whose decisive fragments $C$ satisfy the following ten conditions.

\begin{enumerate}\vs{-2}
  \setlength{\topsep}{-2mm}%
  \setlength{\itemsep}{1mm}%
  \setlength{\parskip}{0cm}%


\item[(1)] The unlinked version of $C$ must be a tree.

\item[(2)] At each layer of $C$, there is at most one cascading block of length $\geq2$ and, if any, it must contain the leftmost node at the   layer.

\item[(3)] In the unlinked, leaf-free version of $C$, every path from each bottom node to the root must pass through at least one cascading block.

\item[(4)] The root of $C$ is the minimum common ancestor of two AND$_{(\omega)}$ gates of a certain cascading block of $C$.

\item[(5)] For any two linked AND$_{(\omega)}$ gates in $C$, their minimum common ancestor must be either an AND gate or an AND$_{(\omega)}$ gate.

\item[(6)] In any cascading block $D$ of $C$, all AND$_{(\omega)}$ gates of $D$ have leafless fan-in at most $2$.

\item[(7)] Every AND$_{(\omega)}$ gate $g$ in $C$ has at most two outputs. If $g$ has exactly two outputs, then one of them must be connected to a CAND gate in a certain cascading block.

\item[(8)] For any two cascading blocks $A$ and $B$ of $C$, if $A$ is an ancestor of $B$, then $A$ and $B$ must be orderly.

\item[(9)] For any two cascading blocks $A$ and $B$ of $C$, if $A$ is an ancestor of $B$, then the block-bottom layer number of $A$ is larger than the block-top layer number of $B$.

\item[(10)] At any layer $\ell$ between the block-top layer number and the block-bottom layer number of a cascading block $D$ of $C$ having  cascading length $\geq2$, there is no gate that is located on the right side of $D$ at layer $\ell$ and has leafless fan-in $\geq2$.
\end{enumerate}
\end{definition}

We succinctly call by a \emph{$k$-cascading semi-circuit} a cascading semi-circuit whose cascading length is at most $k$.

\begin{definition}
A \emph{basic subcircuit} refers to a tree-form subcircuit composed only of OR and AND gates. A \emph{cascading circuit} is built from a basis subcircuit by connecting its leaves to cascading semi-circuits. A \emph{$k$-cascading circuit} refers to a cascading circuit whose cascading semi-circuits all have cascading lengths at most $k$.
\end{definition}

Fig.~\ref{fig:circuit-simulation} illustrates a simple cascading circuit composed of two cascading semi-circuits (namely, the subcircuits rooted at gates $2$ and $5$) connected to a basis subcircuit consisting only of gates $1$, $2$, and $5$.

The \emph{alternation} of a circuit is the maximum number of times when the gate types switch between AND, AND$_{(\omega)}$, and OR in any path from the output gate to an input gate. The \emph{size} of a circuit is the total number of gates and wires (which connect between gates) in it.

Hereafter, we are interested in families of cascading circuits indexed by natural numbers. To discuss such circuit families, we usually demand an appropriate uniformity condition. Of various uniformity notions in the literature, e.g., \cite{Ruz81}, we here use the ``logspace'' uniformity. For any family $\CC=\{C_n\}_{n\in\nat}$ of Boolean circuits, $\CC$   is said to be \emph{logspace-uniform} if there exists a DTM $M$ with a read-only input tape, a rewritable index tape, multiple work tapes, and a write-once output tape such that $M$ takes input of the form $1^n$ and produces $\pair{C_n}$ in $n^{O(1)}$ runtime using $O(\log{n})$ work space, where $\pair{C_n}$ is a standard binary encoding of $C_n$ \cite{Ruz81}.

\begin{definition}
The notation $\mathrm{CCIRcasc,\!alt,\!size}(k,d(n),s(n))$ expresses the collection of all languages solvable by logspace-uniform families of semi-unbounded fan-in $k$-cascading circuits of at most $d(n)$ alternations and at most $s(n)$ size.
\end{definition}

\subsection{Relationship between Auxiliary SNAs and Cascading Circuit Families}\label{sec:simulation}

Sudborough \cite{Sud78} characterized $\logcfl$ in terms of  polynomial-time log-space auxiliary pushdown automata. Founded on Ruzzo's ATM characterization of $\logcfl$ \cite{Ruz80}, Venkateswaran \cite{Ven91} further established a close tie between auxiliary pushdown automata and  semi-unbounded fan-in circuit families.
More precisely, polynomial-time log-space auxiliary pushdown automata have the same computational power as logspace-uniform families of semi-unbounded fan-in circuits of logarithmically-bounded alternation and polynomial size;  in other words, $\logcfl$ coincides with $\mathrm{SAC}^1$.
This current work expands his result to aux-$k$-sna's and semi-unbounded fan-in $k$-cascading circuits.

In this work, we present an exact characterization of aux-$k$-sna's in terms of $k$-cascading circuits by demonstrating direct simulations between these computational models.

\begin{theorem}\label{circuit-character}
For any positive integer $k$, $\auxsna\mathrm{depth,\!space,\!time}(2k,O(\log{n}),n^{O(1)})_{FBS}$ coincides with
$\mathrm{CCIRcasc,\!alt,\!size}(k,O(\log{n}),n^{O(1)})$.
\end{theorem}

Toward the proof of Theorem \ref{circuit-character},  we split the proposition into two independent assertions, which are further generalized as two lemmas: Lemmas \ref{aux-sna-to-circuit} and \ref{circuit-to-aux-sna}.
The first lemma asserts the conversion of aux-$k$-sna's into $k$-cascading circuit families whereas the second lemma asserts the conversion of $k$-cascading circuit families into  aux-$k$-sna's with (frozen) blank sensitivity.

\begin{lemma}\label{aux-sna-to-circuit}
Let $k\geq1$. Let $s,t$ be any bounding functions. For any weakly depth-susceptible aux-$2k$-sna $M$ with right-turn restricted and no frozen blank turn running in time $t(n)$ using space $s(n)$, there exists a logspace-uniform family of semi-unbounded fan-in $k$-cascading circuits of alternation $O(s(n))$ and size $O(t(n))\cdot 2^{O(s(n))}$ that solves $L(M)$.
\end{lemma}


\begin{lemma}\label{circuit-to-aux-sna}
Let $k\geq1$. Let $d,s$ be any bounding functions satisfying $d(n),s(n)\geq n$ for all $n\in\nat$. For any uniform family of semi-unbounded $k$-cascading circuits of alternation at most $d(n)$ and size at most $s(n)$, there exists a (frozen) blank-sensitive aux-$2k$-sna  $M$ that simulates it using $O(d(n)+\log{s(n)})$ space in $O(s(n))$ time.
\end{lemma}

Meanwhile, we postpone the proofs of Lemmas \ref{aux-sna-to-circuit} and \ref{circuit-to-aux-sna} and return to Theorem \ref{circuit-character}. Remember that Section \ref{sec:lemma-proof} will provide the proof of Lemma \ref{aux-sna-to-circuit} and Section \ref{sec:lemma-proof-II} will do the proof of Lemma \ref{circuit-to-aux-sna}.
By combining these generalized lemmas, we instantly obtain the theorem by the following argument.

\begin{proofof}{Theorem \ref{circuit-character}}
To show that $\auxsna\mathrm{depth,\!space,\!time}(k,O(\log{n}),n^{O(1)})_{FBS}$ is included in $\mathrm{CCIRcasc,\!alt,\!size}(k,O(\log{n}),n^{O(1)})$, let us consider an arbitrary language $L$ in $\auxsna\mathrm{depth,\!space,\!time}(k,O(\log{n}),n^{O(1)})_{FBS}$ and take any (frozen) blank-sensitive aux-$k$-sna $M$ that recognizes $L$ in $n^{O(1)}$ time and using $O(\log{n})$ space.
Lemma \ref{relax-FBS} ensures that $M$ can be further assumed to be weakly depth-susceptible and make no frozen blank turn. By Lemma \ref{aux-sna-to-circuit}, we can simulate $M$ using a logspace-uniform family of semi-unbounded fan-in $k$-cascading circuits of depth $O(\log{n})$ and size $n^{O(1)}$. Therefore, $L$ belongs to $\mathrm{CCIRcasc,\!alt,\!size}(k,O(\log{n}),n^{O(1)})$.

Next, we intend to show that $\mathrm{CCIRcasc,\!alt,\!size}(k,O(\log{n}),n^{O(1)})$ is included in  $\auxsna\mathrm{depth,\!space,\!time}(k,O(\log{n}),n^{O(1)})_{FBS}$. Consider an arbitrary language $L$ in $ \mathrm{CCIRcasc,\!alt,\!size}(k,O(\log{n}),n^{O(1)})$. There is a logspace-uniform family $\{C_n\}_{n\in\nat}$ of semi-unbounded fan-in $k$-cascading circuits of size $n^{O(1)}$ and depth $O(\log{n})$ for $L$.
By Lemma \ref{circuit-to-aux-sna}, this circuit family can be simulated by an appropriate aux-$k$-sna with (frozen) blank sensitivity running in $n^{O(1)}$ time using $O(\log{n})$ space. This implies that $L$ is in $\auxsna\mathrm{depth,\!space,\!time}(k,O(\log{n}),n^{O(1)})_{FBS}$.
\end{proofof}

\subsection{Proof of Lemma \ref{aux-sna-to-circuit}}\label{sec:lemma-proof}

In Section \ref{sec:simulation}, we have proven Theorem   \ref{circuit-character} with the help of two supporting lemmas, Lemmas \ref{aux-sna-to-circuit} and \ref{circuit-to-aux-sna}.
This subsection provides the missing proof of Lemma  \ref{aux-sna-to-circuit}.
To begin with the actual proof, let us focus on a weakly depth-susceptible  aux-$k$-sna $M$ with runtime bound $t(n)$ and work space bound $s(n)$ and assume that $M$ is right-turn restricted and makes no frozen blank turn on all inputs.
Recall from Section \ref{sec:useful} a tree-like structure, which represents a computation path of $M$ on input $x$. By Lemma \ref{tree-like-realizable}, it suffices to ``describe'' such a tree-like structure using an appropriate logspace-uniform family $\CC=\{G_n\}_{n\in\nat}$ of $k$-cascading circuits.


\begin{figure}[t]
\centering
\includegraphics*[height=5.0cm]{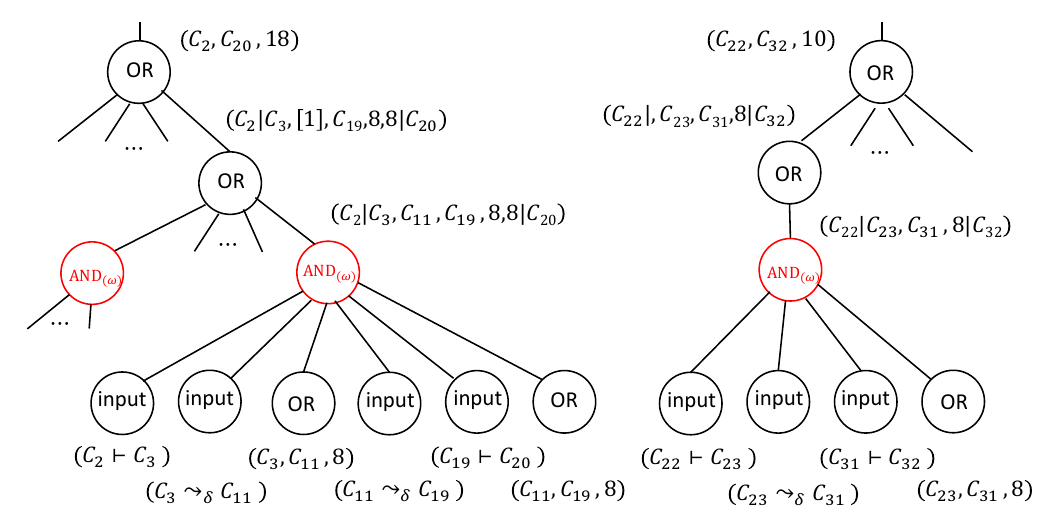}
\caption{Part of a circuit construction, where  $C_t$ denotes a surface configuration at section time $t$ described in Fig.~\ref{fig:storage-tape-head-move}. The AND$_{(\omega)}$ gate with label $(C_2\mmid C_3,C_{11},C_{19},8,8\mmid C_{20})$ has no link to/from any AND$_{(\omega)}$ gate. This situation corresponds to the node $(C_3,C_{11},C_{19},8,8)$ in Fig.~\ref{fig:computation-tree}. Input gates with labels ``$(C_i\leadsto_{\delta}C_j)$'', ``$(C_i\vdash C_j)$'', and ``$(C_i:RT)$'' take the truth values of ``$C_i\leadsto_{\delta}C_j$'', ``$C_i\vdash C_j$'', and ``$C_i$ is a right turn'', respectively.}\label{fig:circuit-description-01}
\end{figure}



\begin{figure}[t]
\centering
\includegraphics*[height=5.2cm]{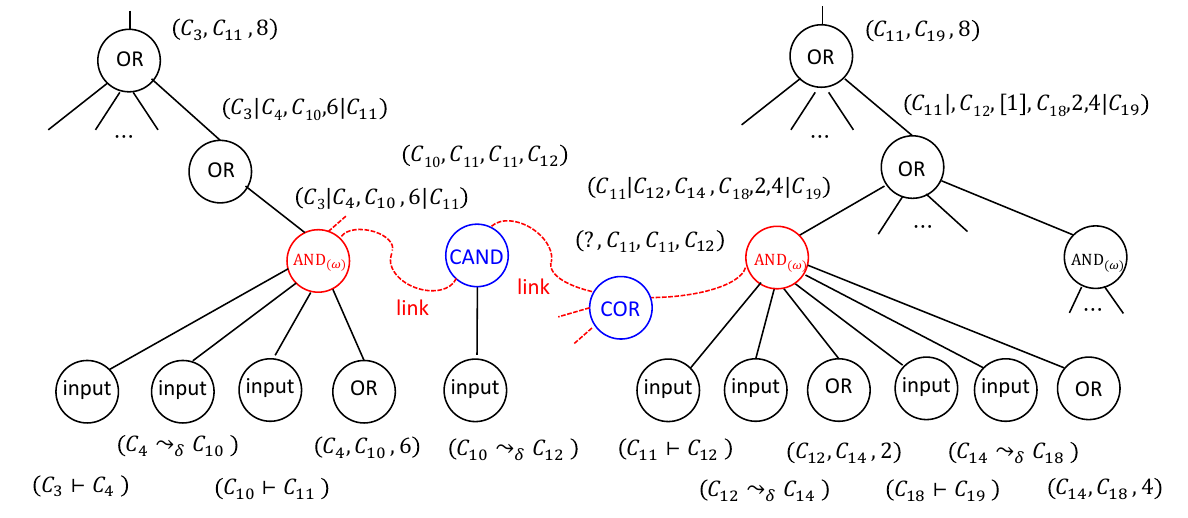}
\caption{Part of a circuit construction concerning linking, where  $C_t$ denotes a surface configuration at section time $t$ in Fig.~\ref{fig:storage-tape-head-move}. The two top OR gates are the same as the bottom OR gates in Fig.~\ref{fig:circuit-description-01}. Due to the space limitation, all  AND$_{(\omega)}$, COR, and CAND gates in a cascading block are intentionally depicted in the same  level.
The terminology
``COR'' refers to OR located in a cascading block.
Because the corresponding tape cell is not yet frozen, the link from the AND$_{(\omega)}$ gate labeled $(C_{3}\mmid C_{4},C_{10},6\mmid C_{11})$, for example, to the AND$_{(\omega)}$ gate with label $(C_{11}\mmid C_{12},C_{14},C_{18},2,4 \mmid C_{19})$ via a CAND gate and a COR gate corresponds to the green dotted line between $(C_4,C_{10},6)$ and $(C_{12},C_{14},C_{18},2,4)$ depicted in Fig.~\ref{fig:computation-tree}.}\label{fig:circuit-description-02}
\end{figure}



\begin{figure}[t]
\centering
\includegraphics*[height=4.8cm]{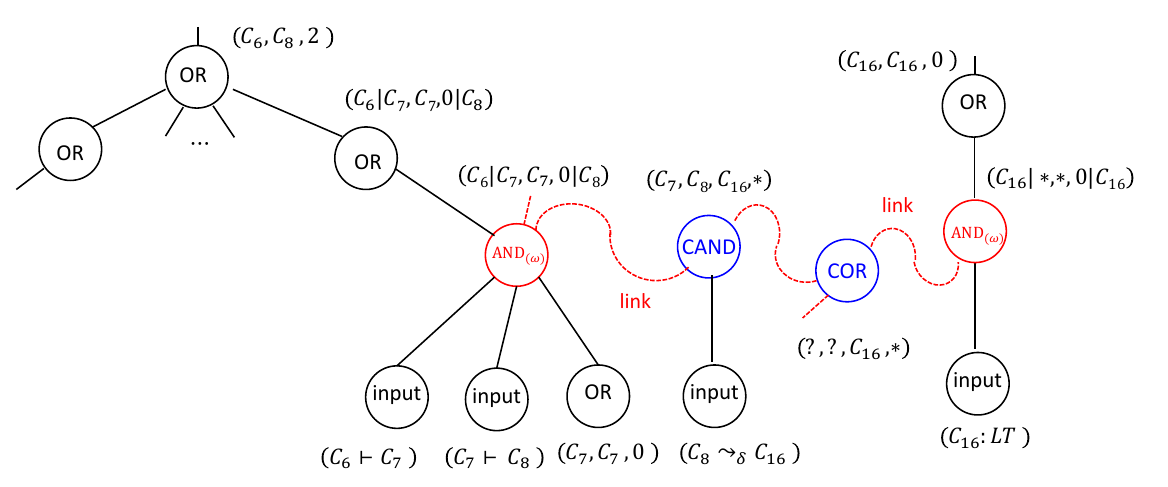}
\caption{Part of a circuit construction. The AND$_{(\omega)}$ gate labeled $(C_6\mmid C_7,C_7,0\mmid C_8)$ is linked to the AND$_{(\omega)}$ gate labeled $(C_{16}\mmid *,*,0\mmid C_{16})$ via a COR gate and a CAND gate.  This link corresponds to the green dotted line between $(C_6,C_8,2)$ and $(C_{16},C_{16},0)$ in Fig.~\ref{fig:computation-tree}.}\label{fig:circuit-description-third}
\end{figure}


Hereafter, we arbitrarily fix an input $x$ and set $n=|x|$. To simplify our construction of $G_n$, described below, we assume that $M$'s storage-tape head must return to the start cell (i.e., cell $0$) when $M$ halts. This can be achieved by moving the tape head leftward by at most $t(n)$ extra tape-head moves.
This implies that $\mathrm{Time}_M(x)\leq 4t(|x|)$. For simplicity, let  $t'(n)$ denote $4t(n)$ and let $SC$ stand for the set of all surface configurations of $M$ on $x$.

All examples given in the following argument correspond to  Fig.~\ref{fig:storage-tape-head-move} and Fig.~\ref{fig:computation-tree}--\ref{fig:circuit-description-02}.


\ms

(I)
To obtain the desired circuit $G_n$, we first construct a ``skeleton'' of  $G_n$ by ignoring cascading blocks and later we will modify it to its complete form by adding missing cascading blocks. Therefore,
this skeleton lacks CAND gates as well as COR gates. These missing gates will be introduced later in (II).
The circuit $G_n$ is constructed inductively layer by layer from its root to leaves in the following fashion.

Let us recall that each computation path of $M$ on $x$ can be depicted in terms of its corresponding tree-like structure. As a quick example, a tree-like structure of Fig.~\ref{fig:computation-tree} represents a computation path of Fig.~\ref{fig:storage-tape-head-move}, which is produced by the movement of a storage-tape head.
We fix a computation path of $M$ on $x$. The subsequent construction of $G_n$ is intuitively intended to represent this tree-like structure, say, $\TT$.

The output gate of $G_n$ takes inputs from OR gates whose labels have the form  $(P_0,P_{acc},i)$ for all numbers $i\in[\mathrm{Time}_{M}(x)]$, where $P_0$ is the initial surface configuration of $M$ on the input $x$ and $P_{acc}$ denotes a ``unique'' accepting surface configuration of $M$ on $x$.
Note that, if $M$ halts in an accepting state along the computation path, then the configuration duo  $(P_0,P_{acc},\mathrm{Time}_{M}(x))$ is realizable if $P_{acc}$ corresponds to the accepting configuration of the computation path. In our example, if a movement of a tape head in Fig.~\ref{fig:storage-tape-head-move} ends in an accepting configuration, then $(C_0,C_{34},34)$ in Fig.~\ref{fig:computation-tree} is the configuration duo $(P_0,P_{acc},\mathrm{Time}_{M}(x))$.


As labels of input gates of $G_n$, we use labels of the form $(P\leadsto_{\delta}R)$, $(P\vdash R)$, $(P:LT)$, and $(P:RT)$, which evaluate the truth values of the corresponding statements: ``$P\leadsto_{\delta}R$'', ``$P\vdash R$'',
``$P$ is a left turn'', and ``$P$ is a right turn'', respectively.
Moreover, we assign the following labels to all other gates.
Each OR gate has a label of the form $(P,R,t)$, $(P'\mmid P_1,P_2,t\mmid R')$, and $(P'\mmid P_1,[1],P_2,t_1,t_2\mmid R')$ for $P,R,P_1,P_2\in SC$, $P',R'\in SC\cup\{*\}$, and $t\in[0,\mathrm{Time}_{M}(x)]_{\integer}$.
Each of AND and AND$_{(\omega)}$  gates has a label of the form $(P\mmid P_1,P_2,t\mmid R)$ and $(P\mmid P_1,P_2,P_3,t_1,t_2\mmid R)$ for $P_1,P_2\in SC$, $P,R\in SC\cup\{*\}$, and $t_1,t_2\in[0,\mathrm{Time}_{M}(x)]_{\integer}$.
Note that the maximum bit size to ``express'' these labels is $O(\log{n})$ since each surface configuration can be expressed using $O(\log{n})$ bits and $k$ is a constant.
Since $P$ and $R$ are allowed to take the value of $*$, it is convenient for us to define $depth(*)$ to be $k$ in the following argument.

In what follows, we wish to construct a segment of $G_n$, in particular,  between two OR gates whose labels are configuration duos of the form $(P,R,t)$.
The labels of these OR gates directly correspond to configuration duos that appear in the tree-like structure $\TT$.
Hence, our goal is to construct an appropriate subcircuit, which fills the gap between those two consecutive configuration duos.
Assume by induction hypothesis that we have already built an OR gate $g$ labeled with $(P,R,t)$.

(1) Let us consider the case of $P=R$ and $t=0$. This case directly corresponds to a leaf of $\TT$. A quick example of such a leaf  in Fig.~\ref{fig:computation-tree} is $(C_7,C_7,0)$ as well as $(C_{16},C_{16},0)$.

(i) We begin with the case of $depth(P)=k$.  This implies that $M$'s storage-tape head should make a left turn. However, this is a contradiction because $M$ makes no frozen blank turn.

(ii) We then assume that $depth(P)<k$. We first attach the AND$_{(\omega)}$ gate of label $(P\mmid *,*,0\mmid R)$ to $g$. To this AND$_{(\omega)}$ gate, we further attach the input gate labeled $(P:LT)$. In Fig.~\ref{fig:circuit-description-third}, since $depth(C_{18})<4$, the OR gate labeled $(C_{16},C_{16},0)$ is directly connected to the AND$_{(\omega)}$ gate labeled $(C_{16}\mmid *,*,0\mmid C_{16})$, which has the input gate labeled $(C_{16}:LT)$ as a child node.

(2) Next, we consider the case of $P\neq R$ and $t\geq2$. An example of such a case of the form $(P,R,t)$ is the node labeled $(C_2,C_{20},18)$ in Fig.~\ref{fig:computation-tree}.
To the gate $g$, we connect a number of OR gates labeled $(P\mmid P_1,P_2,t-2\mmid R)$ and $(P\mmid P_1,[1],P_2,t_1,t_2\mmid R)$  for all possible $P_1$ and $P_2$ of the same depth in $SC$, where  $t_1,t_2\in\nat^{+}$ and $t_1+t_2=t-2$.
We call each of these OR gates by $h$.
In what follows, we discuss two subcases associated
with the choice of $h$.

(i) Consider the case where $h$ has a label $(P\mmid P_1,P_2,t-2\mmid R)$.
To this gate $h$, we attach  a unique  AND$_{(\omega)}$ gate of label $(P\mmid P_1,P_2,t-2\mmid R)$.
As a quick example, in Fig.~\ref{fig:circuit-description-02}, $h$ is the OR gate labeled $(C_{3}\mmid C_4,C_{10},6\mmid C_{11})$, which has the AND$_{(\omega)}$ gate with the label $(C_3\mmid C_4,C_{10},6\mmid C_{11})$ as a child node.
We further demand that, whenever $t=2$, $P_1=P_2$ must follow because this is a leaf of $\TT$, and thus $M$ should make a left turn.
In Fig.~\ref{fig:circuit-description-third}, this special case of $h$ corresponds to the OR gate labeled $(C_6\mmid C_7,C_7,0\mmid C_8)$, which has the AND$_{(\omega)}$ gate with label $(C_6\mmid C_7,C_7,0\mmid C_8)$ as a child.
We call by $f$ this unique AND$_{(\omega)}$ gate.

(a)
Assume that $P_1$ (and thus $P_2$) is frozen blank (i.e., $P_1$'s storage-tape cell contains the frozen blank symbol $B$).
In Fig.~\ref{fig:circuit-description-02}, $f$ is the AND$_{(\omega)}$ gate labeled $(C_3\mmid C_4,C_{10},6\mmid C_{11})$ such that it further has the OR gate labeled $(C_{4},C_{10},6)$ as well as the input gates of label  $(C_{3}\vdash C_{4})$ as its children.
An example of $f$ in Fig.~\ref{fig:circuit-description-01} is the AND$_{(\omega)}$ gate with label $(C_{22}\mmid C_{23},C_{31},8\mmid C_{32})$, which has the OR gate of label $(C_{23},C_{31},8)$.

(b) In the case where $P_1$ is not frozen blank, we attach to $f$ the OR gate labeled $(P_1,P_2,t-2)$ alone (without any input gate).
To the gate $f$, we further attach the OR gate whose  label is of the form  $(P_1,P_2,t-2)$ as well as the input gates with labels of the form $(P\vdash P_1)$ and $(P_2\vdash R)$.

We remark that $(P,R,t)$ is realizable if $(P_1,P_2,t-2)$ is realizable, all input gates are true, and, moreover, linked COR gates are all true.

(ii) Assume that the label of $h$ is $(P\mmid P_1,[1],P_2,t_1,t_2\mmid R)$.  To the gate $h$, we attach a number of AND$_{(\omega)}$ gates of labels $(P\mmid P_1,P',P_2,t_1,t_2\mmid R)$ for all possible $P'\in SC$ of the same depth. This corresponds to the situation that $M$ makes a right turn.
We call such an AND$_{(\omega)}$ gate by $f$.
In our example in Fig.~\ref{fig:circuit-description-01},  if the label of $h$ is $(C_2\mmid C_3,[1],C_{19},8,8\mmid C_{20})$, then  the AND$_{(\omega)}$  gate of label $(C_2\mmid C_3,C_{11},C_{19},8,8\mmid C_{20})$ is connected to this $h$.

Furthermore, we attach two OR gates labeled $(P_1,P',t_1)$ and $(P',P_2,t_2)$ together with the input gates of labels $(P\vdash P_1)$ and $(P_2\vdash R)$. For example in Fig.~\ref{fig:circuit-description-02}, the AND$_{(\omega)}$ gate labeled $(C_{11}\mmid C_{12},C_{14},C_{18},2,4\mmid C_{19})$ is connected to two OR gates of labels $(C_{12},C_{14},2)$ and $(C_{14},C_{18},4)$.


Assume that $\{(P'_j,P'_{j+1},t_j)\mid j\in[m]\}$ is a set of labels of all OR gates generated above. If all configuration duos in this set are  realizable and all other generated input gates are true, then $(P,R,t)$ is also realizable by Lemma \ref{realizable}.

\ms

This is the end of the construction of the skeleton of $G_n$.

\ms

(II)
Finally, we  inductively modify the above-constructed skeleton of $G_n$ and  introduce appropriate CAND and COR gates to connect between two linked AND$_{(\omega)}$ gates, say, $g_1$ and $g_2$, each of which has one of the labels of the form: $(P\mmid P_1,P_2,t\mmid R)$ and $(P\mmid P_1,P',P_2,t_1,t_2\mmid R)$.

(1) We first consider the case where $g_1$ has label $(P\mmid P_1,P_2,t\mmid R)$ and $g_2$ has either $(P'\mmid P'_1,P'_2,t'\mmid R')$ or $(P'\mmid P'_1,P'_2,P'_3,t'_1,t'_2\mmid R')$ as its label.
To this gate $g_1$, we attach an CAND gate whose label is $(P_2,R,P',P'_1)$, which has an input gate, as its child, with label $(P_2\leadsto_{\delta}P'_1)$ if $R=P'$, and with label $(R\leadsto_{\delta}P')$ if $R\neq P'$.
To the gate $g_2$, we further attach a COR gate with label $(?,R,P',P'_1)$ if $R=P'$, and with label $(?,?,P',P'_1)$ if $R\neq P'$.
This case is exemplified in Fig.~\ref{fig:circuit-description-02} with two AND$_{(\omega)}$ gates labeled $(C_{3}\mmid C_{4},C_{10},6\mmid C_{11})$ and $(C_{11}\mmid C_{12},C_{14},C_{18},2,4\mmid C_{19})$. Between them, there are the CAND gate labeled $(C_{10},C_{11},C_{11},C_{12})$ and the COR gate labeled $(?,C_{11},C_{11},C_{12})$.

(2) Consider the case where $g_1$'s label is $(P\mmid P_1,P_2,P_3,t_1,t_2\mmid R)$ and $g_2$'s label is $(P'\mmid P'_1,P'_2,t'\mmid R')$. To the gate $g_1$, we attach a CAND gate labeled $(P_3,R,P',P'_1)$, which further has an input gate, as a child, with label $(P_3\leadsto_{\delta}P'_1)$ if $R=P'$, and with label $(R\leadsto_{\delta}P')$ if $R\neq P'$. To the gate $g_2$, we attach the COR gate with label $(?,R,P',P'_1)$ if $R=P'$, and with label $(?,?,P',P'_1)$ if $R\neq P'$.
An example of this case is two AND$_{(\omega)}$ gates with labels $(C_6\mmid C_7,C_7,0\mmid C_8)$ and $(C_{16}\mmid *,*,0\mmid C_{16})$ in Fig.~\ref{fig:circuit-description-third}. There are the CAND gate labeled $(C_7,C_8,C_{16},*)$ and the COR gate labeled $(?,?,C_{16},*)$ between those two AND$_{(\omega)}$ gates.


Since a COR gate is an OR gate and a CAND gate is an AND gate, if the AND$_{(\omega)}$ gate labeled $g_1$ is true and all input gates attached to its associated CAND gates are true, then the AND$_{(\omega)}$ gate labeled $g_2$ is true.


\ms

Since there are only $2^{O(s(n))}$ configurations of $M$ and $\mathrm{Time}_M(x)\leq t'(n)$, the maximum number of ORs with labels of the form $(P,R,i)$ along any path from the root to a leaf is $O(t(n))\cdot 2^{O(s(n))}$. We also note that $G_n$ is semi-unbounded because of the bounded fan-in of AND and AND$_{(\omega)}$ gates and no consecutive use of them. Thus, the size of $G_n$ is upper-bounded by $O(t(n))\cdot 2^{O(s(n))}$.

Each cascading block reflects links among configuration duos on the same tape cell number of the tree-like structure. For example, the cascading block of  Fig.~\ref{fig:circuit-description-02} corresponds to the tape cell number 3 of the tree-like structure of Fig.~\ref{fig:computation-tree}.
Thus, the cascading length is upper-bounded by $k$.

We then argue that $\CC=\{G_n\}_{n\in\nat}$ correctly solves the language $L(M)$ and $n=|x|$. Assume that $x\in L(M)$. Take an accepting computation path of $M$ on $x$ and fix it.
Along this path, we need to claim  that $G_n(x)$ outputs $1$. This claim can be proven by induction on the construction process of $G_n$ from the computation path of $M$ on $x$ since the construction process precisely expresses the tree-like structure obtained along the computation path.

This completes the proof of Lemma \ref{sec:lemma-proof}.

\subsection{Proof of Lemma \ref{circuit-to-aux-sna}}\label{sec:lemma-proof-II}

In this subsection, we provide the
proof of Lemma \ref{circuit-to-aux-sna} to complete the proof of Theorem  \ref{circuit-character}.

We begin with the lemma's proof by taking an arbitrary logspace-uniform family $\CC=\{G_n\}_{n\in\nat}$ of semi-unbounded fan-in $k$-cascading circuits.
Since $\CC$ is logspace-uniform, there exists a log-space DTM $N$ that produces the binary encoding $\pair{G_n}$ of $G_n$ from the input $1^n$ for each index $n\in\nat$.
We wish to define an aux-$k$-sna $M$ so that it simulates the behaviors of each circuit $G_n$ on all inputs of length $n$. In particular, $M$'s storage tape will be used to simulate the evaluation of a cascading block.

As the label of a gate, for convenience, we use a positive integer. In what follows, for simplicity, we often identify a gate with its label. To store the information on each gate onto the storage tape, we use the following encoding scheme of gate information.
For a given gate $g$ whose children are gates $h_1,h_2,\ldots,h_m$ drawn in order from left to right, we define its encoding $\pair{g,h_1,h_2,\ldots,h_m}$ to be the string  $\pair{g,h_1}\pair{g,h_1}\cdots \pair{g,h_m}$ expressing the label of $g$, where $\pair{g,h_i}=\pair{g}\dollar\pair{h_i}\dollar\sigma\#$ for $\sigma\in\{0,1,2\}$ and $\pair{g}$ and $\pair{h_i}$ are the binary representations of labels of $g$ and $h_i$, respectively, and $\dollar$ and $\#$ are designated symbols not in $\{0,1,2\}$.
The last sign $\sigma$ indicates the gate type of OR, AND, and AND$_{(\omega)}$ if $\sigma$ is $0$, $1$, and $2$, respectively.
If there is no cascading block, then we can evaluate the outcome of $G_n$ in a standard way of performing depth-first search as in, e.g., \cite{Ven91}.


Let us consider an arbitrary gate of $G_n$. We employ a \emph{modified depth-first search} as an underlying strategy to evaluate the outcome of $G_n$ on length-$n$ input $x$.
We inductively follow paths of $G_n$ from the root to leaves and vice versa by storing the information on AND$_{(\omega)}$ gates onto blocks of cells of the storage tape.
Since there are polynomially many gates in $G_n$, the labels of gates of $G_n$ are expressed in $O(\log{s(n)})$ bits, and thus we can store them using a block of $O(\log{s(n)})$ consecutive cells of the storage tape as if this block is a single cell by sweeping the block of cells from left to right and from right to left, as noted in Section \ref{sec:k-sda's}. In the following description of $M$, the phrase ``cell block'' collectively refers to all these consecutive tape cells used to store the information on the label of each gate in $G_n$.


Remember that a cascading circuit can be split into a number of disjoint cascading semi-circuits hung on to a basic subcircuit, which is a simple  tree-form subcircuit built by AND or OR gates only. Let $C_n$ denote any cascading semi-circuit of $G_n$.
Hereafter, we will explain how to evaluate each gate in $C_n$.
To help understand the following evaluation procedure, we intend to exemplify each evaluation step for $C_n$ using the circuit example of  Fig.~\ref{fig:circuit-simulation}.
In particular, we focus on the subcircuit rooted at gate 2 of Fig.~\ref{fig:circuit-simulation}, which forms a $2$-cascading semi-circuit. In our evaluation procedure, we use an $O(\log{n})$ space-bounded auxiliary tape to store the current gate label and search the encoding $\pair{G_n}$ written on an input tape for its connecting gates. We begin with gate $2$ and continue traversing gates $3$ and $6$. During this process, we record the label of gate 2 together with gate $4$ on a storage tape and move to its left child, gate 3. This record will be used later to traverse a subcircuit rooted at gate $4$.
We further record the label of gate 3 together with gate $7$ and move to its left child, gate 6.
If we nondeterministically choose gate 23, then we follow a corresponding link to gate 11. Here, we record the label of gate 11 as well as gate $6$ (which is needed for a backtracking purpose). We then continue traversing through gates 26 to gate 31. If the value of gate 31 is ``true'', then we follow a link to gate 26 and record its label. From gate 31, we start backtracking upward to gate 3 as follows. We first evaluate gate 23. If its value is ``true'', then we revise the record on gates 23 and 6 by erasing the record of gate 6 and move to gate 3. On reaching gate 3, we erase the record on gates $3$ and $7$ written on the storage tape and then start traversing from gate 7 to gate 11. At gate 11, we read the record on gate 11 from the storage tape to check that we indeed reach gate 11. If this is the case, then we evaluate gate 11 and start backtracking upward to gate 3 and to gate 2. We then erase the record on gates $2$ and $4$. From gate 2, we traverse through gates $4$, $8$, $12$, $24$, and $27$.
It is important to note that there is no gate between gate $8$ and gate $28$ having leafless fan-in $\geq2$.
Since the modification of the auxiliary tape is strictly prohibited by the weak depth-susceptibility, we move the storage-tape head to a non-frozen-blank tape cell and start simulating gates $12$, $24$, and $27$.
When we reach gate 28, we retrieve the record on gate 28 and the value of gate 31 from the storage tape and evaluate the value of gate 28.
We then backtrack upward to gate 2 to complete the evaluation of the target subcircuit $C_n$.

In what follows, we intend to formalize the above evaluation procedure.

\ms

(1) Initially, we set $g$ to be the root (i.e., the output gate) of $C_n$. As the first move of $M$, we store (the label of) the root onto the initially-blank cell block of the storage tape and move the storage-tape head to the right.

(2) Assume that $g$ is an AND gate (of fan-in $2$ and fan-out $1$), which is not in any cascading block. There are exactly two incoming gates connected directly to $g$. Let $p_1$ and $p_2$ be these two gates appearing in $G_n$ in order from left to right. In Fig.~\ref{fig:circuit-simulation}, at gates 2 and 3, we record them onto the storage tape.

(a) If this is the first time visiting $g$, then we write $(g,p_2)$ onto an  initially-blank cell block. We then reset $g$ to be $p_1$ and move to the right for the next inductive step. This process is intended to evaluate the outcome of a subcircuit rooted at $p_1$.

(b) If we reach $g$ on the way of backtracking upward to the root, then we access the current cell block to retrieve the gate information of the form $(g,p_2)$. On reaching $g$, we reset $g$ to be $p_2$, delete the record on $(g,p_2)$ by replacing it by a series of $B$s, make a turn, and move forward to $p_2$ so that we will try to evaluate the outcome of a subcircuit rooted at $p_2$. Backtracking made after processing $p_2$ indicates that a subcircuit rooted at $g$ is already evaluated.

(3) Assume that $g$ is an OR gate, which is not in any cascading block but is  directly connected to input gates, OR gates, AND gates, and AND$_{(\omega)}$ gates. Let $h_1,h_2,\ldots,h_m$ denote all those gates, enumerated from left to right. In this case, we nondeterministically choose one of them, say, $h_i$, reset $g$ to be this $h_i$, and move forward.

(4) Assume that $g$ is an AND$_{(\omega)}$  gate (of fan-in $\leq c$ and fan-out $\geq1$).
It is possible to check whether or not the current AND$_{(\omega)}$ gate is in a certain cascading block of length $\geq2$. When it is in a cascading block of length $1$, nevertheless, we do not need to record its information onto the storage tape because it has leafless fan-in $1$ and fan-out $1$ and all connecting input gates can be evaluated directly from the input given to $C_n$. Thus, there is no need to read any record of an AND$_{(\omega)}$ gate from the storage tape.
By this moment, the record on a top AND$_{(\omega)}$ gate of a cascading block has been erased from the storage tape and only $B$s are seen on the storage tape.
It is important to note that we should avoid any movement of input-tape and auxiliary-tape heads while reading $B$ on the storage tape. For this purpose, we first move the storage-tape head to a non-$B$ tape cell before the start of evaluation of each gate and start evaluating the gate by making storage stationary moves.
Next, we consider the case where $g$ is in a certain cascading block of length $\geq2$, say, $D$.

(a) Assume that $g$ is the bottom AND$_{(\omega)}$ gate of $D$. If $g$ is visited for the first time, then we  nondeterministically choose a CAND gate, say, $h_1$ attached to $g$ and then move the storage-tape head rightward to $\Box$. There are a unique COR gate, say, $h_2$ connected directly from $h_1$ and another AND$_{(\omega)}$ gate, say, $g'$ connected from $h_2$. We then evaluate an input gate attached to $h_1$. If it is evaluated to be ``false'', then we immediately reject. Now, we assume that all input gates are evaluated to be ``true''. We first write $(f,g')$ on the storage tape, where $f$ is the gate visited just before reaching $g$. This information on $f$ is needed for backtracking the path that we have taken so far.
If $g$ is visited during the backtracking, then we read $(f,g')$ from the storage tape to find $f$ to return. We also change $(f,g')$ to $(g')$.

(b) Assume that $g$ is the top AND$_{(\omega)}$ gate of $D$. If $g$ is visited for the first time, then we read the information $(g',B)$ from the storage tape. We reject if $g\neq g'$. Otherwise, we evaluate all input gates connecting $g$ to evaluate $g$, change $(g',B)$ to frozen blank, and move to the connecting gate. If $g$ is visited during the backtracking, then we pass through $g$.

(c) Otherwise, let $g$ denote the $i$th AND$_{(\omega)}$ gate, say, $g_i$  from the bottom AND$_{(\omega)}$ gate but neither the bottom nor top AND$_{(\omega)}$ gate.
Firstly, we check if the record written on the storage tape is exactly  $(g_i)$. If not, we reject. Otherwise, we compute the $(i+1)$th AND$_{(\omega)}$ gate, say, $g_{i+1}$ connecting from $g$ and record $(f,g_{i+1})$ on the storage tape if all attached input gates to $g$ are evaluated to be ``true'', where $f$ is the gate visited just before reaching $g$. We then move to the next gate.

(5) Assume that $g$ is a leaf (i.e., an input gate) of $C_n$. Note that $g$ is of the form either $x$ or $\overline{x}$ for an input variable $x$. Thus, we can evaluate $g$ using the input given to $C_n$. If $g$ is evaluated to be ``false'', then we immediately enter a rejecting state and halt since this does not contribute to making $C_n$ output ``true''.
Otherwise, we force the storage-tape head to make a left turn and keep moving to the left as we backtrack a path of $C_n$ toward the root.

(6) In the case where the storage-tape head returns to the start cell, we accept if the root of $C_n$ is evaluated to be ``true'', and we reject otherwise.

\ms

The above simulation nondeterministically checks all paths of the cascading semi-circuit $C_n$ from its leaves to the root and evaluates all gates correctly.

By the above construction together with $N$, we conclude that $M$ is well-defined. Since $C_n$ has size at most $s(n)$, $M$ runs in time $O(s(n))$. The space usage of the auxiliary (work) tape of $M$
is $O(\log{s(n)})$. Notice that the content of each storage-tape cell is modified at most $k$ times and, after the $k$th modification, the cell becomes frozen forever.

Next, we will explain how to simulate the cascading circuit $G_n$, which may consist of two or more cascading semi-circuits. More precisely, we assume that $G_n$ is made up of a basis subcircuit, say, $B_n$ and a number of cascading semi-circuits, say, $C^{(1)},C^{(2)},\ldots,C^{(m)}$ attached to $B_n$. Note that these subcircuits $C^{(i)}$s are disjoint in $G_n$. To make our simulation simpler, we first modify $B_n$ to be \emph{leveled} by inserting extra OR gates of fan-in $1$ so that all leaves, which are connected to the roots of $C^{(1)},C^{(2)},\ldots,C^{(m)}$ of $B_n$, lie on the same level.
We first simulate $B_n$ by running depth-first search using a storage tape by recording the information on the choice of inputs to each AND gate whenever we encounter an AND gate. If we reach a leaf $g$ of $B_n$, then we start the aforementioned evaluation procedure to evaluate the cascading semi-circuit rooted at $g$.
Since any two cascading semi-circuits are disjoint, after each cascading semi-circuit is simulated and evaluated, we move the storage-tape head away to a new $\Box$, place a specific ``marker'' to indicate the start of a new simulation area, and begin a simulation of another cascading semi-circuit using only the tape area on the right side of the marker. In the end, we correctly simulate $G_n$.

\section{In Connection to P}\label{sec:P-connection}

In Section \ref{sec:circuit-family}, we have argued how families of polynomial-size semi-unbounded fan-in $k$-cascading circuits with logarithmic alternations characterize (frozen) blank-sensitive aux-$2k$-sna's running in polynomial time using only $O(\log{n})$ work space. Here, we turn our attention to $\p$ and argue how to relate cascading circuit families to $\p$. To generalize cascading circuits, we intend to remove the upper bounds of  cascading length and alternations of cascading circuit families. We then obtain a cascading-circuit characterization of $\p$.

\begin{proposition}\label{P-cascading-circuit}
The complexity class $\p$ coincides with $\mathrm{CCIRcasc,\!alt,\!size}(n^{O(1)}, n^{O(1)}, n^{O(1)})$.
\end{proposition}

\begin{yproof}
It is well-known that uniform families of polynomial-size Boolean circuits precisely characterize $\p$.
Firstly, we show that  $\p\subseteq \mathrm{CCIRcasc,\!alt,\!size}(n^{O(1)}, n^{O(1)}, n^{O(1)})$. This is obvious because a standard Boolean circuit is also a cascading circuit with no cascading block.

Secondly, we show that $\mathrm{CCIRcasc,\!alt,\!size}(n^{O(1)}, n^{O(1)}, n^{O(1)}) \subseteq \p$. Given a family $\{G_n\}_{n\in\nat}$ of $n^{O(1)}$-cascading circuits of $n^{O(1)}$ alternations and $n^{O(1)}$ size. We intend to evaluate the outcome of $G_n$ by simulating the behavior of $G_n$ on an appropriate DTM $M$ using $n^{O(1)}$ space and $n^{O(1)}$ time.
Since $G_n$ has at most $p(n)$ gates, we list up all these gates onto a work tape of a DTM in order. We start with input gates of $G_n$ and evaluate $G_n$ layer by layer. We write their values on a work tape as follows.
Note that AND gates are viewed as a special form of AND$_{(\omega)}$ gates. We thus evaluate each AND$_{(\omega)}$ gate to be true if all inputs to this gate are ``true'', and we evaluate each OR gate to be true if at least one of inputs to this gate is true. After all gates are evaluated, we check whether the output gate of $G_n$ is true. Since the above procedure can be carried out in polynomial time, $M$ is a polynomial-time DTM simulating $\{G_n\}_{n\in\nat}$.
\end{yproof}

An immediate consequence of Proposition \ref{P-cascading-circuit} is an upper-bound on the computational complexity of $\auxsna\mathrm{depth,\!space,\!time}(k,O(\log{n}),n^{O(1)})_{FBS}$.

\begin{corollary}\label{upepr-bound-P}
For any constant $k\in\nat^{+}$, $\auxsna\mathrm{depth,\!space,\!time}(k,O(\log{n}),n^{O(1)})_{FBS} \subseteq \p$.
\end{corollary}

\begin{yproof}
For simplicity, we write $\LL(k)$ for $\auxsna\mathrm{depth,\!space,\!time}(k,O(\log{n}),n^{O(1)})_{FBS}$.
Firstly, we consider the even case $2k$ for an arbitrary integer $k\geq1$. By Proposition \ref{circuit-character},  $\LL(2k)$ coincides with $\mathrm{CCIRcasc,\!alt,\!size}(k, O(\log{n}), n^{O(1)})$. By Proposition \ref{P-cascading-circuit}, we conclude that $\LL(2k)$ is included in $\p$. Secondly, we consider the odd case $2k+1$. Since $\LL(2k)\subseteq \LL(2k+1) \subseteq \LL(2k+2)$, it follows that $\LL(2k+1)$ is also included in $\p$.
\end{yproof}

In the end, we remark that, since $\mathrm{LOG}(k\mathrm{SNA}_{FBS})\subseteq \auxsna\mathrm{depth,\!space,\!time}(k,O(\log{n}),n^{O(1)})_{FBS}$ holds by a simple simulation, Corollary \ref{upepr-bound-P} yields an upper bound of the computational complexity of $\mathrm{LOG}(k\mathrm{SNA}_{FBS})$ by $\p$.

\section{Brief Discussions and Future Research Directions}

We have studied aux-$k$-sna's, which are nondeterministic variants of aux-$k$-sda's proposed in \cite{Yam21}, and we have presented its characterization in terms of uniform families of polynomial-size semi-unbounded fan-in Boolean circuits with logarithmic alternations, called \emph{$k$-cascading circuits} built up from length-$k$ cascading blocks, composed of semi-unbounded fan-in gates together with AND$_{(\omega)}$ gates of bounded fan-in and unbounded fan-out.

For future research, we wish to list eight key open questions associated with aux-$k$-sna's and $k$-cascading circuits.

\renewcommand{\labelitemi}{$\circ$}
\begin{enumerate}
  \setlength{\topsep}{-2mm}%
  \setlength{\itemsep}{1mm}%
  \setlength{\parskip}{0cm}%

\item To achieve our characterization, in essence, we have required  two specific conditions of aux-$k$-sna's: weak depth-susceptibility and no frozen blank turn. As shown in Lemma \ref{relax-FBS}, these conditions are immediately induced from another notion of (frozen) blank sensitivity. We now wonder if these conditions can be further relaxed or eliminated completely. If so, the elimination of the conditions further simplifies the model of the aux-$k$-sna's.

\item The computational model of cascading circuits is a simple modification of the standard circuit model by supplementing cascading blocks. We dare to ask whether or not there are any ``natural'' computational models, completely different from cascading circuits, which precisely capture the behaviors of aux-$k$-sna's.

\item Our circuit characterization works for $\auxsna\mathrm{depth,\!space,\!time}(2k,O(\log{n}),n^{O(1)})_{FBS}$  with an arbitrary positive integer $k$. How can we modify the definition of cascading circuits to establish an exact characterization of  $\auxsna\mathrm{depth,\!space,\!time}(k,O(\log{n}),n^{O(1)})_{FBS}$ for any $k\in\nat^{+}$ in terms of circuit families?

\item The probabilistic variant of Hibbard's limited automata was discussed extensively in \cite{Yam19}. It is of importance to look into a similar probabilistic variant of $k$-sna's and explore its characteristic features for a better understanding of storage automata.

\item In Section \ref{sec:P-connection}, the removal of the upper bounds of cascading length and alternations significantly increases the computational complexity of cascading circuits. If we expand only the cascading length from $k$ to $O(\log^t{n})$, where $t$ is  any constant in $\nat^{+}$, then what is the computational complexity of $\mathrm{CCIRcasc,\!alt,\!size}(O(\log^t{n}),O(\log{n}),n^{O(1)})$?

\item We have introduced $\mathrm{LOG}k\mathrm{SNA}$ based on $k$-sna's.  Is it true that  $\mathrm{LOG}k\mathrm{SDA} \neq  \mathrm{LOG}k\mathrm{SNA}$ and $\mathrm{LOG}k\mathrm{SNA} \neq  \mathrm{LOG}(k+1)\mathrm{SNA}$ for any number $k\in\nat^{+}$?

\item Does $\auxsna\mathrm{depth,\!space,\!time}(k,O(\log{n}),n^{O(1)})_{FBS}$ coincide with $\mathrm{LOG}k\mathrm{SNA}_{FBS}$? If so, can we expand it to the case of $\mathrm{LOG}k\mathrm{SDA}$?

\item In \cite[Lemma 2.1]{Yam21}, $\dcfl$ is characterized by depth-susceptible 2-sda's, whereas we characterize $\dcfl$ using (frozen) blank-sensitive 2-sda's in Proposition \ref{new-result-FBS}. What is the relationship between these two notions over $k$-sda's for any $k\geq3$? More specifically, is it true that $k\mathrm{SDA}_{suseptible} = k\mathrm{SDA}_{FBS}$ for any $k\geq3$? Here, the subscript ``susceptible'' refers to the condition that underlying $k$-sda's are all depth-susceptible.

\item An additional memory device of ``counter'' is appended to pushdown automata in \cite{Yam23,Yam24} to enhance the computational power of the pushdown automata. Does the use of multiple counters help increase  the computational power of $k$-sna's?
\end{enumerate}



\let\oldbibliography\thebibliography
\renewcommand{\thebibliography}[1]{%
  \oldbibliography{#1}%
  \setlength{\itemsep}{-2pt}%
}
\bibliographystyle{plain}

\end{document}